\documentclass[10pt,twoside]{article}
\usepackage[utf8]{inputenc}
\usepackage[hmargin=1in,vmargin=1in]{geometry}

\usepackage{jeffe} 

\usepackage[dvipsnames,usenames]{xcolor}
\usepackage{hyperref}
\hypersetup{colorlinks=true, urlcolor=Blue, citecolor=Green, linkcolor=BrickRed, breaklinks, unicode}

\usepackage{graphicx}
\graphicspath{{figures/}}
\DeclareGraphicsExtensions{.pdf,.png,.jpg}

\usepackage[charter]{mathdesign}
\usepackage{berasans,beramono}
\usepackage{microtype}

\usepackage[nocompress]{cite}
\usepackage{amsmath,mathtools,stmaryrd}
\usepackage{enumitem}
\usepackage{footmisc}

\usepackage{atbegshi}
\AtBeginDocument{\AtBeginShipoutNext{\AtBeginShipoutDiscard}}

\def\etal{\textit{et~al.}}
\def\polylog{\mathop{\mathrm{polylog}}}
\def\eps{\varepsilon}

\def\reals{\mathbb{R}}

\def\ceil#1{\lceil #1 \rceil}
\def\seq#1{\langle #1 \rangle}
\def\set#1{\{ #1 \}}
\def\abs#1{\mathopen| #1 \mathclose|}		
\def\norm#1{\mathopen\| #1 \mathclose\|}	

\def\Seq#1{\left\langle #1 \right\rangle}
\def\Set#1{\left\{ #1 \right\}}

\def\Paren#1{\left( #1 \right)}		
\def\Brack#1{\left[ #1 \right]}		

\def\arcto{\mathord\shortrightarrow}

\def\Vor{\text{Vor}}

\allowdisplaybreaks

\makeatletter
\long\def\@makecaption#1#2{
    \vskip 10pt
    \setbox\@tempboxa\hbox{{\footnotesize {\bf #1.} #2}}
    \ifdim \wd\@tempboxa >\hsize         
        {\footnotesize {\bf #1.} #2\par}
      \else                              
        \hbox to\hsize{\hfil\box\@tempboxa\hfil}
    \fi}
\makeatother

\def\EMPH#1{\textbf{\boldmath #1}}
\def\n@te#1{\textsf{\boldmath \textbf{$\langle\!\langle$#1$\rangle\!\rangle$}}\leavevmode}
\def\note#1{\textcolor{red}{\n@te{#1}}}

\usepackage{textcomp}
\def\cost{\$}
\def\dist{\delta}

\def\nn{\operatorname{nn}}
\def\Center{\operatorname{center}}

\newtheorem{lemma}{Lemma}[section]
\newtheorem{theorem}[lemma]{Theorem}

\newtheorem{remark}[lemma]{Remark}
\numberwithin{figure}{section}

\begin{document}
\begin{titlepage}

\title{Clustering under Perturbation Stability in Near-Linear Time}%

\thanks{Pankaj Agarwal has been partially supported by NSF grants IIS-18-14493 and CCF-20-07556. Kamesh Munagala is supported by NSF grants CCF-1637397 and IIS- 1447554; ONR award N00014-19-1-2268; and DARPA award FA8650-18-C-7880.}

\author{
    Pankaj K.\ Agarwal\thanks{Department of Computer Science, Duke University, USA.}
    \qquad
    Hsien-Chih Chang\thanks{Department of Computer Science, Dartmouth, USA.}
		\qquad
		Kamesh Munagala\footnotemark[2]
		\medskip \\
		Erin Taylor\footnotemark[2]
    \qquad
    Emo Welzl\thanks{Department of Computer Science, ETH Zürich, Switzerland.}
}

\date{September 28, 2020}

\maketitle

\begin{abstract}
We consider the problem of center-based clustering in low-dimensional Euclidean spaces under the perturbation stability assumption.
An instance is $\alpha$-stable if the underlying optimal clustering continues to remain optimal even when all pairwise distances are arbitrarily perturbed by a factor of at most $\alpha$.
Our main contribution is in presenting efficient exact algorithms for $\alpha$-stable clustering instances whose running times depend near-linearly on the size of the data set when $\alpha \ge 2 + \sqrt{3}$.
For $k$-center and $k$-means problems, our algorithms also achieve polynomial dependence on the number of clusters, $k$, when $\alpha \ge 2 + \sqrt{3} + \eps$ for any constant $\eps > 0$ in any fixed dimension. For $k$-median, our algorithms have polynomial dependence on $k$ for $\alpha > 5$ in any fixed dimension; and for $\alpha \ge 2 + \sqrt{3}$ in two dimensions.
Our algorithms are simple, and only require applying techniques such as local search or dynamic programming to a suitably modified metric space, combined with careful choice of data structures.
\end{abstract}


\setcounter{page}{0}
\thispagestyle{empty}
\global\let\newpagegood\newpage
\global\let\newpage\relax
\end{titlepage}
\global\let\newpage\newpagegood
\pagestyle{myheadings}
\markboth{P.\ K.\ Agarwal, H.-C. Chang, K.\ Munagala, E.\ Taylor, and E.\ Welzl}
		{Clustering under Perturbation Stability in Near-Linear Time}



\section{Introduction}
Clustering is a fundamental problem in unsupervised learning and data summarization, with wide-ranging applications that span myriad areas.
Typically, the data points are assumed to lie in a Euclidean space, and the goal in center-based clustering is to open a set of $k$ centers to minimize the objective cost, usually a function over the distance from each data point to its closest center.
The $k$-median objective minimizes the sum of distances; the $k$-means minimizes the sum of squares of distances; and the $k$-center minimizes the longest distance.
In the worst case, all these objectives are NP-hard even in 2D~\cite{ms-ccglp-1984,mnv-pkpnh-2012}.

A substantial body of work has focused on developing polynomial-time approximation algorithms and analyzing natural heuristics for these problems.
Given the sheer size of modern data sets, such as those generated in genomics or mapping applications, even a polynomial-time algorithm is too slow to be useful in practice---just computing all pairs of distances can be computationally burdensome.
What we need is an algorithm whose running time is near-linear in the input size and polynomial in the number of clusters.

Because of NP-hardness results, we cannot hope to compute an optimal solution in polynomial time, but in the worst case an approximate clustering can be different from an optimal clustering.
We focus on the case when the optimal clustering can be recovered under some reasonable assumptions on the input that hold in practice.
Such methodology is termed ``beyond worst-case analysis'' and has been adopted by recent work~\cite{angel2017local,afshani2017instance,bilu2012stable}.
In recent years, the notion of \emph{stability} has emerged as a popular assumption under which polynomial-time optimal clustering algorithms have been developed.
An instance of clustering is called \emph{stable} if any ``small perturbation''  of input points does not change the optimal solution.
This is natural in real datasets, where often, the optimal clustering is clearly demarcated, and the distances are obtained heuristically.
Different notions of stability differ in how ``small perturbation'' is defined, though most of them are related.
In this paper, we focus on the notions of stability introduced in
Bilu and Linial~\cite{bilu2012stable}
and Awasthi, Blum, and Sheffet~\cite{abs-cbcps-2012}.
A clustering instance is \emph{$\alpha$-perturbation resilient} or \emph{$\alpha$-stable} if the optimal clustering does not change when all distances are perturbed by a factor of at most $\alpha$.
Similarly, a clustering instance is \emph{$\alpha$-center proximal} if any point is at least a factor of $\alpha$ closer to its own optimal center than any other optimal center.
Awasthi, Blum, and Sheffet showed that $\alpha$-stability implies $\alpha$-center proximity~\cite{abs-cbcps-2012}.
This line of work designs algorithms to recover the \emph{exact} optimal clustering---the ground truth---in polynomial time for $\alpha$-stable instances.

This paper also focuses on recovering the optimal clustering for stable clustering instances.
But instead of focusing on polynomial-time algorithms and optimizing the value of $\alpha$, we ask the question:
\emph{Can algorithms be designed that compute exact solutions to stable instances of Euclidean center-based clustering that run in time near-linear in the input size?}
We note that an $(1+\e)$-approximation solution, for an arbitrarily small constant $\e > 0$, may differ significantly from an optimal solution (the ground truth) even for stable instances, so one cannot hope to use an approximation algorithm to recover the optimal clustering.

\subsection{Our Results}
In this paper, we make progress on the above question, and present near-linear time algorithms for finding optimal solutions of stable clustering instances with moderate values of $\alpha$.
In particular, we show the following meta-theorem:


\begin{theorem}
    \label{Th:meta}
Let $X$ be a set of $n$ points in $\reals^d$ for some constant $d$, let $k \geq 1$ be an integer, and let $\alpha \geq 2 + \sqrt{3}$ be a parameter. If the $k$-median, $k$-means, or $k$-center clustering instance for $X$ is $\alpha$-stable, then the optimal solution can be computed in $\tilde{O}(n \poly k + f(k))$ time.
\end{theorem}

In the above theorem, the $\tilde{O}$ notation suppresses logarithmic terms in $n$ and the spread of the point set. The function $f(k)$ depends on the choice of algorithm, and we present the exact dependence below.
We also omit terms depending solely on the dimension, $d$.
Furthermore, the above theorem is robust in the sense that the algorithm is not restricted to choosing the input points as centers (\emph{discrete setting}), and can potentially choose arbitrary points in the Euclidean plane as centers (\emph{continuous setting}, sometimes referred to as the \emph{Steiner point setting})---indeed, we show that these notions are identical under a reasonable assumption on stability. 

At a more fine-grained level, we present several algorithms that require mild assumptions on the stability condition. 
In the results below, as well as throughout the paper, we present our results both for the Euclidean plane, as well as generalizations to higher (but fixed number of) dimensions.



\begin{description}
\item[Dynamic Programming.] In Section~\ref{S:dp}, we present a dynamic programming algorithm that computes the optimal clustering in $O(nk^2 + n \polylog n)$ time for $\alpha$-stable $k$-means, $k$-median, and $k$-center in any fixed dimension,
provided that $\alpha \ge 2 + \sqrt{3} + \eps$ for any constant $\eps > 0$.  For $d = 2$, it suffices to assume that $\alpha \geq 2 + \sqrt{3}$.
\item[Local Search.] In Sections~\ref{S:local-search} and~\ref{S:efficient-cost}, we show that the standard $1$-swap local-search algorithm, which iteratively swaps out a center in the current solution for a new center as long as the resulting total cost improves, computes an optimal clustering for $\alpha$-stable instances of $k$-median assuming $\alpha > 5$.  We also show that it can be implemented in
$O(nk^2 \log^3 n\log \Delta)$ for $d = 2$ and in $O(n k^{2d - 1} \polylog n \log \Delta)$ for $d > 2$; $\Delta$ is the spread of the point set.\footnote{The \emph{spread} of a point set is the ratio between the longest and shortest pairwise distances.}
\item[Coresets.] In the Section~\ref{S:coresets}, we use multiplicative coresets to compute the optimal clustering for $k$-means, $k$-median and $k$-center in any fixed dimension, when  $\alpha \ge 2+\sqrt{3}$.
The running time is $O(nk^2 + f(k))$ where $f(k)$ is an exponential function of $k$.
\end{description}

\begin{remark} While the current analysis of the dynamic programming based algorithm suggests that it is better than the local-search and coreset based approaches, the latter are of independent interest---our local-search analysis is considerably simpler than the previous analysis~\cite{fks-ealbs-2019}, and coresets have mostly been used to compute approximate, rather than exact, solutions.
We also note that our analysis of the local-search algorithm is probably not tight. Furthermore, variants of all three approaches might work for smaller values of $\alpha$.
We note that the value of $\alpha$ assumed in the above results in larger than what is known for polynomial-time algorithm (e.g.\ $\alpha \ge 2$ in Angelidakis \etal~\cite{amm-aspp-2017}) and that in some applications the input may not satisfy our assumption, but our results are a big first step toward developing near-linear time algorithms for stable instances.
We are not aware of any previous near-linear time algorithms for computing optimal clustering even for larger values of $\alpha$.
We leave the problem of reducing the assumption on $\alpha$ as an important open question.
\end{remark}

\subparagraph*{Techniques.} The key difficulty with developing fast algorithms for computing the optimal clustering is that some clusters could have a very small size compared to others. This issue persists even when the instances are stable.
Imagine a scenario where there are multiple small clusters, and an algorithm must decide whether to merge these into one cluster while splitting some large cluster, or keep them intact.
Now imagine this situation happening recursively, so that the algorithm has multiple choices about which clusters to recursively split.
The difference in cost between these options and the size of the small clusters can be small enough that any $(1+\e)$-approximation can be agnostic, while an exact solution cannot.
As such, work on finding exact optima use techniques such as dynamic programming~\cite{amm-aspp-2017} or local search with large number of swaps~\cite{coh-faslk-2018,fks-ealbs-2019} in order to recover small clusters.
Other work makes assumptions lower-bounding the size of the optimal clusters or the spread of their centers~\cite{dlv-ekmca-2019}.

Our main technical insight for the first two results is simple in hindsight, yet powerful: For a stable instance, if the Euclidean metric is replaced by another metric that is a good approximation, then the optimal clustering does not change under the new metric and in fact the instance remains stable albeit with a smaller stability parameter.
In particular,
%
we replace the Euclidean metric with an appropriate \emph{polyhedral metric}---that is, a convex distance function where each unit ball is a regular polyhedron---yielding efficient procedures for the following two primitives:

\begin{itemize}
\item \textbf{Cost of $\mathbf{1}$-swap.}  Given a candidate set of centers $S$, maintain a data structure that efficiently updates the total cost if center $x \in S$ is replaced by center $y \notin S$.
\item \textbf{Cost of $\mathbf{1}$-clustering.}  Given a partition of the data points, maintain a data structure where the cost of $1$-clustering (under any objectives) can be efficiently updated as partitions are merged.
\end{itemize}
%
%

We next combine the insight of changing the metrics with additional techniques.
For local search, we build on the approach in~\cite{cs-lssci-2017,coh-faslk-2018,fks-ealbs-2019} that shows local search with $t$-swaps for large enough constant $t$ finds an optimal solution for stable instances in polynomial time for any fixed-dimension Euclidean space.
None of the prior analysis directly extends as is to $1$-swap, which is critical in achieving near-linear running time---note that even when $t = 2$ there is a quadratic number of candidate swaps per step.

For the dynamic programming algorithm, we use the following insight:  In Euclidean spaces, for $\alpha \ge 2 + \sqrt{3}$, the longest edge of the minimum spanning tree over the input points partitions the data set in two, such that any optimal cluster lies completely in one of the two sides of the partition. Combined with the change of metrics one can achieve near-linear running time.


%

%

\subsection{Related Work}
\label{SS:related-work}
All of $k$-median, $k$-means, and $k$-center are widely studied from the perspective of approximation algorithms and are known to be hard to approximate~\cite{feder1988optimal}.
Indeed, for general metric spaces,  $k$-center is hard to approximate to within a factor of $2 - \e$~\cite{HochbaumS}; $k$-median is hard to $(1+2/e)$-approximate~\cite{jms-ngafl-2002}; and $k$-means is hard to $(1+8/e)$-approximate in general metrics~\cite{cohenaddad_et_al2019}, and is hard to approximate within a factor of $1.0013$ in the Euclidean setting~\cite{lee2017improved}.
Even when the metric space is Euclidean, $k$-means is still NP-hard when $k=2$~\cite{dasgupta2008hardness,adhp-nhess-2009}, and there is an $n^{\Omega(k)}$ lower bound on running time for $k$-median and $k$-means in $4$-dimensional Euclidean space under the exponential-time hypothesis~\cite{cmrr-blc-2018}.

There is a long line of work in developing $(1+\e)$-approximations for these problems in Euclidean spaces.
The holy grail of this work has been the development of algorithms that are near-linear time in $n$, and several techniques are now known to achieve this.
This includes randomly shifted quad-trees~\cite{arora1998approximation}, coresets~\cite{agarwal2005geometric,har2004no,har2004coresets,feldman2007ptas,badoiu2002approximate}, 
sampling~\cite{kumar2004simple}, and local search~\cite{coh-faslk-2018,ckm-lsyas-2019,cfs-ntasc-2019}, 
among others.

There are many notions of clustering stability that have been considered in literature~\cite{ostrovsky,bbv-dfcsf-2008,bbg-aca-2009,ab-cts-2009,br-dsccl-2014,akb-cscq-2016,dvw-csiek-2017,abjk-acscq-2018,kumar}.
The exact definition of stability we study here was first introduced in
Awasthi \etal~\cite{abs-cbcps-2012}; their definition in particular resembles the one of
Bilu and Linial~\cite{bilu2012stable} for max-cut problem, which later has been adapted to other optimization problems~\cite{mmssw-cmtsc-2011,mmv-blsim-2014,amm-aspp-2017,bb-nepsg-2017,aabcd-blsca-2018}.
Building on a long line of work~\cite{abs-cbcps-2012,bl-cpr-2016,bhw-kccpr-2016,bc-ptasc-2016}, which gradually reduced the stability parameter, Angelidakis \etal~\cite{amm-aspp-2017} present a dynamic programming based polynomial-time optimal algorithm for discrete $2$-stable instances for all center-based objectives.


Chekuri and Gupta~\cite{cg-prckc-2018} show that a natural LP-relaxation is integral for the $2$-stable $k$-center problem.
Recent work by Cohen-Addad~\cite{cs-lssci-2017} provides a framework for analyzing local search algorithms for stable instances.
This work shows that for an $\alpha$-stable instance with $\alpha > 3$, any solution is optimal if it cannot be improved by swapping $\ceil{2/(\alpha - 3)}$ centers.
%
Focusing on Euclidean spaces of fixed dimensions, Friggstad \etal~\cite{fks-ealbs-2019} show that a local-search algorithm with $O(1)$-swaps runs in polynomial time under a $(1+\delta)$-stable assumption for any $\delta > 0$.
However, none of the algorithms for stable instances of clustering so far have running time near-linear in $n$, even when the stability parameter $\alpha$ is large, points lie in $\reals^2$, and the underlying metric is Euclidean.

On the hardness side, solving $(3 - \delta)$-center proximal $k$-median instances in general metric spaces is NP-hard for any $\delta>0$~\cite{abs-cbcps-2012}.
When restricted to Euclidean spaces in arbitrary dimensions, Ben-David and Reyzin~\cite{br-dsccl-2014} showed that for every $\delta > 0$, solving discrete $(2 - \delta)$-center proximal $k$-median instances is NP-hard.
Similarly, the clustering problem for discrete $k$-center remains hard for $\alpha$-stable instances when $\alpha < 2$, assuming standard complexity assumption that NP $\neq$ RP~\cite{bhw-kccpr-2016}.
Under the same complexity assumption, discrete $\alpha$-stable $k$-means is also hard when $\alpha < 1+\delta_0$ for some positive constant~$\delta_0$~\cite{fks-ealbs-2019}.
Deshpande \etal~\cite{dlv-ekmca-2019} showed it is NP-hard to $(1+\e)$-approximate $(2-\delta)$-center proximal $k$-means instances.

\section{Definitions and Preliminaries}
\label{S:prelim}

\subparagraph*{Clustering.}
Let $X = \{ p_1, \dots , p_n\}$ be a set of $n$ points in $\reals^d$, and let $\dist\colon \reals^d \times \reals^d \to \reals_{\geq 0}$ be a distance function (not necessarily a metric satisfying triangle inequality).
For a set $Y \subseteq \reals^d$, we define $\EMPH{$\dist(p, Y)$} \coloneqq \min_{y \in Y} \dist(p, y)$.
A \EMPH{$k$-clustering} of $X$ is a partition of $X$ into $k$ non-empty \EMPH{clusters} $X_1, \dots, X_k$.
We focus on center-based clusterings that are induced by a set $S \coloneqq \{ c_1, \dots, c_k \}$ of $k$ \EMPH{centers}; each $X_i$ is the subset of points of $X$ that are closest to $c_i$ in $S$ under $\dist$, that is, $X_i \coloneqq \Set{ p \in X \mid \dist(p, c_i) \leq \dist(p, c_j) }$ (ties are broken arbitrarily).
Assuming the nearest neighbor of each point of $X$ in $S$ is unique (under distance function $\dist$), $S$ defines a $k$-clustering of $X$.
Sometimes it is more convenient to denote a $k$-clustering by its set of centers $S$.

The quality of a clustering $S$ of $X$ is defined using a cost function \EMPH{$\cost(X, S)$}; cost function~$\cost$ depends on the distance function $\dist$, so sometimes we may use the notation $\cost_\dist$ to emphasize the underlying distance function.
The goal is to compute $S^* \coloneqq \argmin_{S} \cost(X, S)$ where the minimum is taken over all subsets $S \subset \reals^d$ of $k$ points.
Several different cost functions have been proposed, leading to various optimization problems.
We consider the following three popular variants:


\begin{itemize}\itemsep=0pt
\item \EMPH{$k$-median clustering}: the cost function is $\cost(X, S) = \sum_{p \in X}  \dist(p, S)$.
\item \EMPH{$k$-means clustering}: the cost function is $\cost(X, S) = \smash{\sum_{p \in X} (\dist(p, S))^2}$.
\item \EMPH{$k$-center clustering}: the cost function is $\cost(X, S) = \max_{p \in X}  \dist(p, S)$.
\end{itemize}

In some cases we wish $S$ to be a subset of $X$, in which case we refer to the problem as the \EMPH{discrete $k$-clustering} problem.
For example, the discrete $k$-median problem is to compute
\(
\argmin_{S \subseteq X, |S| = k} \sum_{p \in X}  \dist(p, S).
\)
The discrete $k$-means and discrete $k$-center problems are defined analogously.

Given point set $X$, distance function $\dist$, and cost function $\cost$, we refer to $(X, \dist, \cost)$ as a \EMPH{clustering instance}.
If $\cost$ is defined directly by the distance function $\dist$, we use $(X, \dist)$ to denote a clustering instance.
Note that a center of a set of points may not be unique (e.g.\ when $\dist$ is defined by the $L_1$-metric and $\cost$ is the sum of distances) or it may not be easy to compute (e.g.\ when $\dist$ is defined by the $L_2$-metric and $\cost$ is the sum of distances).



\subparagraph*{Stability.}
Let $X$ be a point set in Euclidean space $\reals^d$.
For $\alpha \ge 1$, a clustering instance $(X, \dist, \cost_\dist)$ is \EMPH{$\alpha$-stable} if for any \emph{perturbed distance function} $\tilde\dist$ (not necessary a metric) satisfying  $\dist(p, q) \leq \tilde\dist(p, q) \leq \alpha \cdot \dist(p,q)$ \text{ for all } $p, q \in \reals^d,$ any optimal clustering of $(X, \dist, \cost_\dist)$
is also an optimal clustering of $(X, \tilde\dist, \cost_{\tilde\dist})$.
Note that the cluster centers as well as the cost of optimal clustering may be different for the two instances.
We exploit the following property of stability, which follows directly from its definition.

\begin{lemma}
\label{L:unique-optimal}
Let $(X, \dist)$ be an $\alpha$-stable clustering instance with $\alpha > 1$.
Then the optimal clustering $O$ of $(X, \dist)$ is unique.
\end{lemma}


\begin{proof}
Assume for contradiction that there are two optimal clusterings $O$ and $O'$.
There must be a point $p$ in $X$ that belongs to a cluster centered at $c$ in $O$ but is assigned to a different center $c'$ in $O'$.
Consider the perturbed distance $\tilde\delta$ by scaling inter-cluster distances in $O$ by an $\alpha$ factor while preserving all intra-cluster distances.  Then
\[
\alpha \cdot \dist(p,c) \le \alpha \cdot \dist(p,c') = \tilde\dist(p,c') \le \tilde\dist(p,c) = \dist(p,c),
\]
where the first inequality is by definition of clustering $O$,
the second inequality is by definition of clustering $O'$ still being optimal under $\tilde\dist$ by $\alpha$-stability, and the two equalities are follows from how the perturbed distance is defined.
This give a contradiction as long as $\alpha > 1$.
\end{proof}

\subparagraph*{Metric approximations.}
The next lemma, which we rely on heavily throughout the paper, is the observation that a change of metric preserves the optimal clustering as long as the new metric is a $\beta$-approximation of the original metric satisfying $\beta < \alpha$.

\begin{lemma}
\label{L:change-metric}
Given point set $X$, let $\dist$ and $\dist'$ be two metrics satisfying $\dist(p, q) \leq \dist'(p, q) \leq \beta \cdot \dist(p, q)$ for all $p$ and $q$ in $X$ for some $\beta$.
Let $(X, \dist)$ be an $\alpha$-stable clustering instance with $\alpha > \beta$.
Then the optimal clustering of $(X, \dist)$ is also the optimal clustering of $(X, \dist')$, and vice versa.
Furthermore, $(X, \dist')$ is an $(\alpha/\beta)$-stable clustering instance.
\end{lemma}


\begin{proof}
Because $(X, \dist)$ is $\alpha$-stable for $\alpha > \beta$, the optimal clustering of $(X, \dist)$ is also an optimal clustering of $(X, \dist')$ by taking $\dist'$ to be the perturbed distance.
Now, for an arbitrary perturbed distance $\tilde\dist'$ satisfying $\dist'(p, q) \leq \tilde\dist'(p, q) \leq (\alpha/\beta) \cdot \dist'(p,q)$ for all $p, q \in X$, one has
\[
\dist(p,q) \le \dist'(p, q) \leq \tilde\dist'(p, q) \leq \frac{\alpha}{\beta} \cdot \dist'(p,q) \le \alpha \cdot \dist(p,q),
\]
and therefore the optimal clustering $O$ of $(X,\dist)$ and $(X,\dist')$ is must be an optimal clustering of $(X,\tilde\dist')$, proving that $(X, \dist')$ is $(\alpha/\beta)$-stable.
Providing $\alpha > \beta$, the optimal clustering of $(X, \dist')$ is again unique by Lemma~\ref{L:unique-optimal}; in other words, the optimal clustering of $(X, \dist')$ is by definition equal to the optimal clustering of $(X,\dist)$.
\end{proof}

%

\subparagraph*{Polyhedral metric.}
In light of the metric approximation lemma, we would like to approximate the Euclidean metric without losing too much stability, using a collection of convex distance functions generalizing the $L_\infty$-metric in Euclidean space.
Let $\EMPH{$N$} \subseteq S^{d-1}$ be a centrally-symmetric set of $\gamma$ unit vectors (that is, if $u \in N$ then $-u \in N$) such that for any unit vector $v \in S^{d - 1}$, there is a vector $u \in N$ within angle $\arccos(1-\e) = O(\sqrt{\e})$.
The number of vectors needed in $N$ is known to be $O(\e^{-(d-1)/2})$.
%
We define the \EMPH{polyhedral metric} $\dist_N \colon \reals^d \times \reals^d \to \reals_{\geq 0}$ to be
$\EMPH{$\dist_N(p, q)$} \coloneqq \max_{u \in N} \seq{ p - q, u }$.


Since $N$ is centrally symmetric, $\dist_N$ is symmetric and thus a metric.
The unit ball under $\dist_N$ is a convex polyhedron, thus the name polyhedral metric.
By construction, an easy calculation shows that for any $p$ and $q$ in $\reals^d$,
\(
\norm{p - q} \geq \dist_N(p, q) \geq (1-\e) \cdot \norm{p - q}.
\)
By scaling each vector in $N$ by a $1+\e$ factor, we can ensure that $(1+\e) \cdot \norm{p - q} \geq \dist_N(p, q) \geq \norm{p - q}$. By taking $\e$ to be small enough, the optimal clustering for $\alpha$-stable clustering instance $(X, \norm{\cdot}, \cost)$ is the same as that for $(X, \dist_N, \cost)$ by Lemma~\ref{L:change-metric}, and the new instance $(X, \dist_N, \cost)$ is $(1-\e)\alpha$-stable if the original instance $(X, \norm{\cdot}, \cost)$ is $\alpha$-stable.

\subparagraph*{Center proximity.}
A clustering instance $(X, \dist)$ satisfies \emph{$\alpha$-center proximity property}~\cite{abs-cbcps-2012} if for any distinct optimal clusters $X_i$ and $X_j$ with centers $c_i$ and $c_j$ and any point $p \in X_i$, one has
\(
\alpha \cdot \dist(p, c_i) < \dist(p, c_j).
\)
Awasthi, Blum, and Sheffet showed that any $\alpha$-stable instance satisfies $\alpha$-center proximity~\cite[Fact~2.2]{abs-cbcps-2012} (also~\cite[Theorem~3.1]{amm-aspp-2017} under metric perturbation).
Optimal solutions of $\alpha$-stable instances satisfy the following separation properties.\footnote{We give an additional list of known separation properties in Appendix~\ref{S:stability-properties}.}

\begin{itemize}\itemsep=0pt
\item $\alpha$-center proximity implies that
\( (\alpha - 1) \cdot \dist(p, c_i) < \dist(p, q) \text{ for any } p \in X_i \text{ and any } q \not \in X_i. \)
For $\alpha \geq 2$, a point is closer to its own center than to any point of another cluster.%
\footnote{They are known as \emph{weak center proximity}~\cite{bhw-kccpr-2016} and \emph{strict separation property}~\cite{bbv-dfcsf-2008,br-dsccl-2014} respectively.\label{F:separation}}
\item  For $\alpha \geq 2 + \sqrt{3}$, $\alpha$-center proximity implies that
 \( \dist(p,p') < \dist(p,q)  \text{ for any }  p,p' \in X_i \text{ and any }  q \not \in X_i. \)
In other words, from any point $p$ in $X$, any \emph{intra}-cluster distance to a point $p'$ is shorter than any \emph{inter}-cluster distance to a point $q$.\footref{F:separation}
\end{itemize}
We make use of the following stronger intra-inter distance property on $\alpha$-stable instances, which allows us to compare \emph{any} intra-distance between two points in $X_i$ and \emph{any} inter-distance between a point in $X_i$ and a point in $X_j$.
\begin{lemma}
\label{L:intra-inter}
    Let $(X,\delta)$ be an $\alpha$-stable instance, $\alpha > 1$, and let $X_1$ be a cluster in an optimal clustering with $q \in X \setminus X_1$ and $p,p',p'' \in X_1$. If $\delta$ is a metric, then $\delta(p,p') \le  \delta(p'',q)$ for $\alpha \ge 2 + \sqrt{5}$. If $\delta$ is the Euclidean metric in $\mathbb{R}^d$, then $\delta(p,p') \le  \delta(p'',q)$ for $\alpha \ge 2 + \sqrt{3}$.
\end{lemma}
\noindent \emph{Proof.} See Appendix~\ref{S:stability-properties}.
\noindent Finally, we note that it is enough to consider the discrete version of the clustering problem for stable instances.

\begin{lemma}
\label{lem:discreteOpt}
For any $\alpha$-stable instance $(X,\dist, \cost_\dist)$ with $\alpha \ge 2+\sqrt{3}$, any continuous optimal $k$-clustering is a discrete optimal $k$-clustering and vice versa.
\end{lemma}


\begin{proof}
Consider $O$ to be an optimal solution of an arbitrary $\alpha$-stable instance $(X,\dist)$ in the continuous setting; denote the centers in $O$ as $o_1, \dots, o_k$.
Define solution $O'$ to be the set of centers $\nn_1, \dots, \nn_k$, where $\nn_i$ is defined to be the nearest point of $o_i$ in $X$.
By definition $O'$ is a discrete solution as all centers $\nn_i$ lie in $X$.
We now argue that $O'$ is in fact an optimal solution of the $k$-clustering instance $(X,\dist)$.

First we show that $\nn_i$ must be a point that was assigned to $o_i$ in clustering $O$.
Assume for contradiction that $\nn_i$ was in a different cluster with center $o_j$.
Let $p$ be an arbitrary point in cluster $O_i$.
By center proximity one has $\dist(p,\nn_i) > (\alpha-1) \cdot \dist(p,o_i)$.
But then this implies $(\alpha-1) \cdot \dist(p,o_i) < \dist(p,\nn_i) \le \dist(p,o_i) + \dist(o_i,\nn_i)$,
that is, $\dist(p,o_i) \le (\alpha-2) \cdot \dist(p,o_i) < \dist(o_i,\nn_i)$ given $\alpha \ge 3$, a contradiction.

Now again take an arbitrary point $p$ in some arbitrary cluster $O_i$.
Compare the distances $\dist(p,\nn_i)$ and $\dist(p,\nn_j)$ for any other center $\nn_j$ in $O'$.
By \cite[Theorem~8]{br-dsccl-2014} for $\alpha > 2 + \sqrt{3}$ any intra-cluster distance is smaller than any inter-cluster distance.
Thus, $\dist(p,\nn_i) < \dist(p,\nn_j)$ since $\nn_i$ lies in $O_i$ and $\nn_j$ lies in $O_j$.
Therefore the clustering formed by the centers in $O'$ is identical to the clustering of $O$, thus proving that $O'$ is an optimal solution of $(X,\dist)$.
\end{proof}



\section{Efficient Dynamic Programming}
\label{S:dp}
We now describe a simple, efficient algorithm for computing the optimal clustering for the $k$-means, $k$-center, and $k$-median problem assuming the given instance is $\alpha$-stable for $\alpha \geq 2 + \sqrt{3}$.
Roughly speaking, we make the following observation: if there are at least two clusters, then the two endpoints of the longest edge of the minimum spanning tree of $X$ belong to different clusters, and no cluster has points in both subtrees of the minimum spanning tree delimited by the longest edge.
We describe the dynamic programming algorithm in Section~\ref{SS:fast-dp} and then describe the procedure for computing cluster costs in Section~\ref{SS:1-clustering}.
We summarize the results in this section by the following theorem.

\begin{theorem}
\label{thm:dp-main}
    Let $X$ be a set with $n$ points lying in $\reals^d$ and $k \geq 1$ an integer.
    If the $k$-means, $k$-median, or $k$-center instance for $X$ under the Euclidean metric is $\alpha$-stable for $\alpha \ge 2 + \sqrt{3} + \eps$ for any constant $\eps > 0$, then the optimal clustering can be computed in $O(nk^2 + n \polylog n)$ time.
    For $d = 2$ the assumption can be relaxed to $\alpha \ge 2 + \sqrt{3}$.
\end{theorem}


\subsection{Fast Dynamic Programming}
\label{SS:fast-dp}

The following lemma is the key observation for our algorithm.

\begin{lemma}
\label{L:longest-edge}
Let $(X, \dist, \cost)$ be an $\alpha$-stable $k$-clustering instance with $\alpha \ge 2 + \sqrt{3}$ and $k \geq 2$, and let $T$ be the minimum spanning tree of $X$ under metric $\dist$. Then
(1) The two endpoints $u$ and $v$ of the longest edge $e$ in $T$ do not belong to the same cluster;
(2) each cluster lies in the same connected component  of $T \setminus \{ e \}$.
\end{lemma}

%
\begin{proof}
Assume for contradiction that the longest spanning tree edge $uv$ belongs to the same cluster $X_i$ in the optimal $k$-clustering $O$.
Since $k>1$, there is at least one other cluster $X_j$ of $O$ with a spanning tree edge $xy$ connecting $X_i$ to $X_j$.
Given $\alpha \ge 2+\sqrt{3}$, $d(u, v) < d(x, y)$ by Lemma~\ref{L:intra-inter}, a contradiction.
The second statement follows from Angelidakis \etal~\cite[Lemma~4.1]{amm-aspp-2017}.
\end{proof}

\subparagraph*{Algorithm.}

We fix the metric $\dist$ and the cost function $\cost$.
For a subset $Y \subseteq X$ and for an integer $j$ between $1$ and $k-1$, let \EMPH{$\mu(Y; j)$} denote the optimal cost of an $j$-clustering on $Y$ (under $\dist$ and $\cost$).
Recall that our definition of $j$-clustering required all clusters to be non-empty, so it is not defined for $|Y| < j$.
For simplicity, we assume that $\mu(Y;j) = \infty$ for $|Y| < j$.
Let \EMPH{$T$} be the minimum spanning tree on $X$ under $\dist$, let $uv$ be the longest edge in $T$; let $X_u$ and $X_v$ be the set of vertices of the two components of $T\setminus \{ uv \}$.
Then $\mu(X; k)$ satisfies the following recurrence relation:

\begin{equation}
\mu(X; k) = \begin{cases} \label{eq:dp}
  \mu(X; 1) & \text{if $k = 1$,} \\
  \infty & \text{if $k > \abs{X}$,} \\
  \min_{1 \leq i < k} \Set{ \mu(X_u; i) + \mu(X_v; k-i) } & \text{if $\abs{X} > 1$ and $k > 1$.}
\end{cases}
\end{equation}

Using recurrence~(\ref{eq:dp}), we compute $\mu(X; k)$ as follows.
Let \EMPH{$\mathsf{R}$} be a \emph{recursion tree}, a binary tree where each node $v$ in $\mathsf{R}$ is associated with a subtree $T_v$ of $T$.
If $v$ is the root of $\mathsf{R}$, then $T_v = T$.
Recursion tree $\mathsf{R}$ is defined recursively as follows.
Let $X_v \subseteq X$ be the set of vertices of $T$ in $T_v$. If $|X_v| = 1$, then $v$ is a leaf.
Each interior node $v$ of $T$ is also associated with the longest edge $e_v$ of $T_v$.
Removal of $e_v$ decomposes  $T_v$ into two connected components, each of which is associated with one of the children of $v$. After having computed $T$, $\mathsf{R}$ can be computed in $O(n \log n)$ time by sorting the edges in decreasing order of their costs.%
\footnote{Tree $\mathsf{R}$ is nothing but the minimum spanning tree constructed by Kruskal's algorithm.}

For each node $v \in \mathsf{R}$ and for every $i$ between $1$ and $k-1$, we compute $\mu(X_v; i)$ as follows.
If $v$ is a leaf, we set $\mu(X_v;1) = 0$ and $\mu(X_v;i) = \infty$ otherwise.
For all interior nodes $v$, we compute $\mu(X_v;1)$ using the algorithms described in Section~\ref{SS:1-clustering}.
Finally, if $v$ is an interior node and $i > 1$, we compute $\mu(X_v;i)$ using the recurrence relation (\ref{eq:dp}).
Recall that if $w$ and $z$ are the children of $v$, then $\mu(X_w; \ell)$ and $\mu(X_z; r)$ for all $\ell$ and $r$ have been computed before we compute $\mu(X_v;i)$.

Let $\tau(n)$ be the time spent in computing $T$ plus the total time spent in computing $\mu(X_v, 1)$ for all nodes $v \in \mathsf{R}$. Then the overall time taken by the algorithm is $O(nk^2 + \tau(n))$.
What is left is to compute the minimum spanning tree $T$ and all $\mu(X_v, 1)$ efficiently.

\subsection{Efficient Implementation}
\label{SS:1-clustering}


In this section, we show how to obtain the minimum spanning tree and compute $\mu(X_v; 1)$ efficiently for $1$-mean, $1$-center, and $1$-median when $X \subseteq \reals^d$.
%
We can compute the Euclidean minimum spanning tree $T$ in $O(n\log n)$ time in $\reals^2$~\cite{toth2017handbook}.
We can then compute $\mu(X_v; 1)$ efficiently either under Euclidean metric (for $1$-mean), or
switch to the $L_1$-metric and compute $\mu(X_v; 1)$ efficiently using Lemma~\ref{L:change-metric} (for $1$-center and $1$-median).

There are two difficulties in extending the 2D data structures to higher dimensions.
No near-linear time algorithm is known for computing the Euclidean minimum spanning tree for $d \geq 3$, and we can work with the $L_1$-metric only if $\alpha \geq \sqrt{d}$ (Lemma~\ref{L:change-metric}).
We address both of these difficulties by working with a polyhedral metric $\dist_N$.
Let $\alpha \ge 2 + \sqrt{3} + \Omega(1)$ be the stability parameter.
By taking the number of vectors in $N$ (defined by the polyhedral metric) to be large enough, we can ensure that $(1 - \e) \norm{p - q} \leq \delta_N(p, q) \leq \norm{p - q}$ for all $p, q \in \reals^d$.
By Lemma~\ref{L:change-metric}, $X$ is an $\alpha$-stable instance under $\delta_N$ for $\alpha \geq 2 + \sqrt{3}$.
We first compute the minimum spanning tree of $X$ in $O(n \polylog n)$ time under $\delta_N$ using the result of Callahan and Kosaraju~\cite{ck-faggp-1993}, and then compute $\mu(X_v, 1)$.


\subparagraph*{Data structure.}
We compute $\mu(X_v;1)$ in a bottom-up manner.
When processing a node $v$ of $\mathsf{R}$, we maintain a dynamic data structure \EMPH{$\Psi_v$} on $X_v$ from which $\mu(X_v;1)$ can be computed quickly.
The exact form of $\Psi_v$ depends on the cost function to be described below.
Before that, we analyze the running time $\tau(n)$ spent on computing every $\mu(X_v;1)$.
Let $w$ and $z$ be the two children of $v$.
Suppose we have $\Psi_w$ and $\Psi_z$ at our disposal and suppose $|X_w| \leq |X_z|$.
We insert the points of $X_w$ into $\Psi_z$ one by one and obtain $\Psi_v$ from which we compute  $\mu(X_v;1)$.
Suppose $Q(n)$ is the update time of $\Psi_v$ as well as the time taken to compute $\mu(X_v;1)$ from $\Psi_v$.
The total number of insert operations performed over all nodes of $\mathsf{R}$ is $O(n \log n)$ because we insert the points of the smaller set into the larger set at each node of $\mathsf{R}$~\cite{ht-fafnc-84, st-dsdt-83}.
Hence $\tau(n) = O(Q(n) \cdot n \log n)$.
We now describe the data structure for each specific clustering problem.

\subparagraph*{$\mathbf{1}$-mean.}

We work with the $L_2$-metric.
Here the center of a single cluster consisting of $X_v$ is the centroid $\sigma_v \coloneqq \Paren{\sum_{p \in X_v} p} / |X_v|$, and
\(
\mu(X_v;1) = \sum_{p \in X_v} \norm{p}^2 - |X_v| \cdot \norm{\sigma_v}^2.
\)
At each node $v$, we maintain $\sum_{p \in X_v} p$ and $\sum_{p \in X_v} \norm{p}^2$.
Point insertion takes $O(1)$ time so $Q(n) = 1$.


\subparagraph*{$\mathbf{1}$-center.}
As mentioned in the beginning of the section, we can work with the $L_1$-metric for $d = 2$.
We wish to find the smallest $L_1$-disc (a diamond) that contains $X_v$.
Let $e^+ = \Paren{1, 1}$ and $e^- = \Paren{-1, 1}$.
Then the radius $\rho_v$ of the smaller $L_1$-disc containing $X_v$ is

\begin{equation}
\rho_v = \frac{1}{2} \max \Set{
    \max_{p \in X_v} \seq{p, e^+} - \min_{p \in X_v} \seq{p, e^+},
    \max_{p \in X_v} \seq{p, e^-} - \min_{p \in X_v} \seq{p, e^-}
}.
\end{equation}

We maintain the four terms $\max_{p \in X_v} \seq{p, e^+}$, $\min_{p \in X_v} \seq{p, e^+}$, $\max_{p \in X_v} \seq{p, e^-}$, and $\min_{p \in X_v} \seq{p, e^-}$ at~$v$.
A point can be inserted in $O(1)$ time and $\rho_v$ can be computed from these four terms in $O(1)$ time.
Therefore, $Q(n) = O(1)$.

\subparagraph*{$\mathbf{1}$-center in higher dimensions}
\label{sp:1-center-highd}
For $d > 2$, we work with a polyhedral metric $\delta_N$ with $N = 2^{O(d)}$.
For a node $v$, we need to compute the smallest ball $B(X_v)$ under $\dist_N$ that contains $X_v$.
We need a few geometric observations to compute the smallest enclosing ball efficiently.

For each $u \in N$, let $H_u$ be the halfspace $\seq{x, u} \leq 1$, that is, the halfspace bounded by the hyperplane tangent to $S^{d - 1}$ at $u$ and containing the origin.
Define $Q \coloneqq \bigcap_{u \in N} H_u$.
A ball of radius $\lambda$ centered at $p$ under $\dist_N$ is $P + \lambda Q$.
For a vector $u \in N$, let $\overline{p}_u \coloneqq \argmax_{p \in x_v} \seq{p, u}$ be the maximal point in direction $u$.
Set $\overline{X}_v \coloneqq \Set{ \overline{p}_u \mid  u \in N }$.
The following simple lemma is the key to computing $B(X_v)$.

\begin{lemma}
\label{lem:highdim-ball}
Any $\dist_N$-ball that contains $\overline{X}_v$ also contains $X_v$.
\end{lemma}

By Lemma~\ref{lem:highdim-ball}, it suffices to compute $B(\overline{X}_v)$.
The next observation is that $B(\overline{X}_v)$ has a basis of size $d+1$, i.e.\ there is a subset $Y$ of $d+1$ points of $X_v$ such that $B(Y) = B(\overline{X}_v) = B(X_v)$.
One can try all possible subsets of $\overline{X}_v$ in $O(\gamma^{d+1}) = 2^{O(d^2)}$ time.%
\footnote{A more complex algorithm can compute $B(\overline{X}_v)$ in $\gamma \cdot 2^{O(d)} = 2^{O(d)}$ time, but we ignore this improvement.}
We note that $\overline{X}_v$ can be maintained under insertion in $O(\gamma) = 2^{O(d)}$ time, and we then re-compute $B(\overline{X}_v)$ in $2^{O(d^2)}$ time.
Hence, $Q(n) = O(1)$.

\subparagraph*{$\mathbf{1}$-median.}
Similar to $1$-center, we work with the polyhedral metric.
Fix a node $v$ of $T$. For a point $x \in \reals^d$, let $F_v(x) = \sum_{p \in X_v} \delta_N(x, p)$ which is a piecewise-linear function.
Our goal is to compute $\xi_v^* = \argmin_{x \in \reals^d} F_v(x)$.
Our data structure is a dynamic range-tree~\cite{agarwal1999geometric} used for orthogonal range searching that can insert a point in $O(\log n)$ time.
Using multi-dimensional parametric search~\cite{agarwal1993ray}, $\xi_v^*$ can be computed in $O(\poly\log n)$ time after each update.

\subparagraph*{$\mathbf{1}$-median in higher dimensions}
\label{sp:appendix-1median}
For simplicity, we describe the data structure for $d = 2$.
It extends to high dimensions in a straightforward manner.

\begin{figure}[htb]
\centering
\includegraphics[scale=0.75, clip]{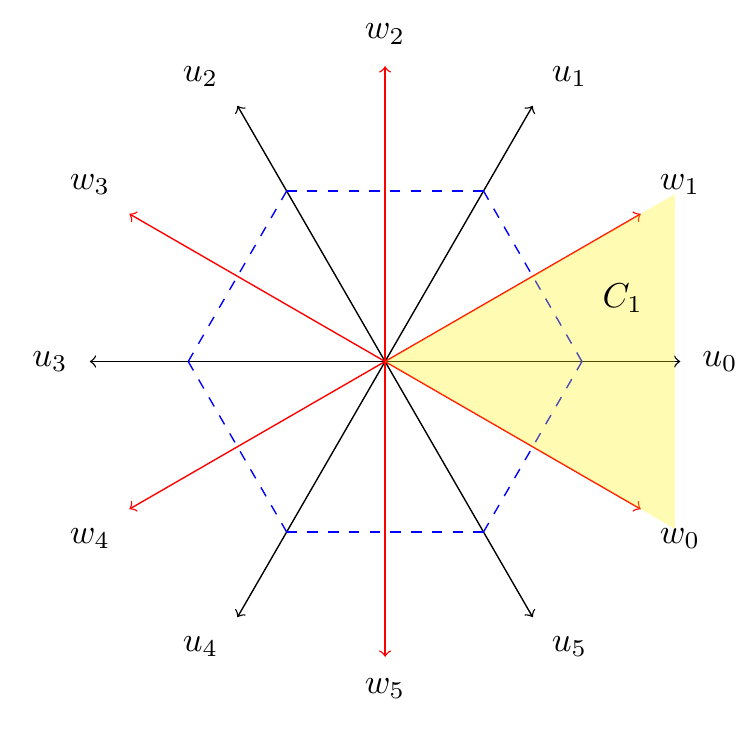}
\caption{Polyhedral metric defined by $N = \{ u_0, \dots u_5 \}$, with $C_1$ corresponding to $\Psi_1$.}
\label{F:1-med_poly}
\end{figure}

Fix a node $v$. We describe the data structure for maintaining $\xi_v^*$ under insertion of points~\footnote{We note that $\xi_v^*$ may not be unique. The minimum may be realized of a convex polygon.}
Let $N = \{ u_0, \dots , u_{r - 1} \} \subset S^1$ be the set of unit vectors that define the metric $\delta_N$.
We partition the plane into a family $\mathcal{C} = \{ C_0, \dots, C_{r - 1} \}$ of $r$ cones such that for a point $p \in C_i$, $\delta_N(p, 0) = \seq{p, u_i}$.
$C_i$ is defined by unit vectors $w_{i - 1}, w_i$, where $w_j$ is the unit vector in direction $(u_{j-1} + u_j) \rvert_2$; see Figure~\ref{F:1-med_poly}.
For a point $x \in \reals^2$ and $j < r$, let $C_j(x) = C_j + x$. Then,

\[\begin{aligned} F_v(x) &= \sum_{i = 1}^{r-1} \sum_{p \in C_i(x) \cap X_v} \delta_N(p, x) = \sum_{i = 0}^{r - 1} \sum_{p \in C_i(x) \cap X_v} \seq{p - x, u_i} \\
&= \sum_{i = 0}^{r - 1} \sum_{p \in C_i(x) \cap X_v} \seq{p, u_i} - \sum_{i = 0}^r |C_i(x) \cap X_v | \cdot \seq{x, u_i}.
\end{aligned} \]

We note that $F_v(x)$ is a piecewise-linear convex function.
We construct a separate data structure $\Psi_i$ for each $i$ so that for any $x \in \reals^2$, $\Psi_i$ computes $\alpha_i(x) = \sum_{p \in C_i(x) \cap X_v} \seq{p, u_i}$ and $\beta_i(x) = |C_i(x) \cap X_v |$.
$\Psi_i$ is basically a dynamic $2D$ range tree in which the coordinates of a point are described with $w_{i - 1}$, $w_i$ as the coordinate axes; see~\cite{de1997computational}.
$\Psi_i$ requires $O(n \log n)$ space, a query can be answered in $O(\log^2 n)$ time, and a point can be inserted in $O(\log^2 n)$ amortized time.
Hence, the overall query and update time is $O(r \log^2 n)$.
We note that $\alpha_i, \beta_i$ for $i < r$ can be used to compute the linear function  $L_{v, x}$ that represents $x$ (recall that $F_v$ is piecewise linear).
Let $\Psi = (\Psi_0, \dots , \Psi_{r - 1})$ denote the overall data structure, and let $Q_0(x)$ be the above query procedure on $\Psi$.

Using $\Psi, Q_0(\cdot)$, and multi-dimensional parametric search, we compute $\xi_v^*$ as follows.
For a line $\ell$ in $\reals^2$, let $\xi_{v, \ell}^* = \argmin_{x \in \ell} F_v(x)$.  We first describe how to compute  $\xi_{v, \ell}^*$.
Let $q$ be a point on $\ell$.
By invoking $Q_0 (q)$ on $\Psi$, we can compute $F_v(q)$ as well as $L_{v, q}$.
Using $L_{v, q}$, we can determine whether $q = \xi_{v, l}^*$, q lies to left of $\xi_{v, l}^*$, or $q$ lies to the right of $\xi_{v, l}^*$.
We refer to this as the ``decision'' procedure.
In order to compute $\xi_{v}^*$, we simulate $Q_0$ generically on $\xi_{v, l}^*$ without knowing its value and using $Q_0$ on known points as the decision procedure at each step of this generic procedure.
More precisely, at each step, a $\Psi_i$ compares the $w_{i - 1}$ or $w_i$-coordinate, say $w_i$-coordinate $\xi_{v, x}$, with a real value $\Delta$.
Let $q$ be the intersection point of  $\ell$ and the line $w_i = \Delta$.
By invoking $Q_0(q)$ on $\Psi$, we can determine in $O( r \log^2 n)$ time whether $q$ lies to the left or right of $\xi_v^*$, which in turn determines whether the $w_i$-coordinate of $\xi_{v, \ell}^*$ is smaller or greater than $\Delta$.
(If $q = \xi_{x, \ell}^*$, then we have found $\xi_{x, \ell}^*$).
Hence, each step of the decision procedure can be determined in $O(r \log^2 n)$ time.
The total time taken by the generic procedure is $O(r^2 \log^4 n)$.
The parametric search technique ensures that the generic procedure will query with $\xi_{v, \ell}^*$ as one of the steps, so the decision procedure will detect this and return $\xi_{v, \ell}^*$.

Let $Q_1 (\ell)$ denote the above procedure to compute $\xi_{v, \ell}^*$.
By simulating $Q_0$ on $\xi_v^*$ generically but now using $Q_1$ as the decision procedure, we can compute $\xi_{v}^*$ in $O(r^3 \log^6 n)$ time.
Hence, we can maintain $\xi_v^*$ in $O(\log^6 n)$ time under insertion of a point.
In higher dimensions, $Q_0$ takes $O(\log^d n)$ time in $\reals^d$. So the parametric search will take $O(\log^{d(d+1)} n)$ time to compute $\xi_v^*$.

\section{{\boldmath $k$}-Median: Single-Swap Local Search}
\label{S:local-search}

%
We customize the standard local-search framework for the $k$-clustering problem~\cite{cs-lssci-2017,ckm-lsyas-2019,frs-lsypk-2019}. In order to recover the optimal solution, we must define near-optimality more carefully.
Let $(X, \dist)$ be an instance of $\alpha$-stable $k$-median in $\reals^2$ for $\alpha > 5$. 
By Lemma~\ref{lem:discreteOpt}, it suffices to consider the \emph{discrete} $k$-median problem
In Section~\ref{SS:local-search}, we describe a simple local-search algorithm for finding the optimal clustering of $(X, \dist)$.
In Section~\ref{SS:optimality} we show that the algorithm terminates within $O(k \log (n\Delta))$ iterations.
%
We obtain the following.
\begin{theorem}
Let $(X,\dist)$ be an $\alpha$-stable instance of the $k$-median problem for some $\alpha > 5$ where $X$ is a set of $n$ points in $\reals^2$ equipped with $L_p$-metric $\dist$.
The $1$-swap local search algorithm terminates with the optimal clustering in $O(k\log(n\Delta))$ iterations.
\end{theorem}


\subparagraph*{Local-search algorithm.}
\label{SS:local-search}
%
The local-search algorithm maintains a $k$-clustering induced by a set $S$ of $k$ cluster centers. At each step, it finds a pair of points $x \in X$ and $y \in S$ such that $\cost(X, S + x - y)$ is minimized.
If $\cost(X, S + x - y) \geq \cost(X, S)$, it stops and returns the $k$-clustering induced by $S$.
Otherwise it replaces $S$ with $S + x - y$ and repeats the above step. The pair $(x, y)$ will be referred to as a \EMPH{$1$-swap}.

\subparagraph*{Local-search analysis.}
\label{SS:optimality}
%
%
The high-level structure of our analysis follows Friggstad \etal~\cite{frs-lsypk-2019}, however new ideas are needed for $1$-swap.
In this subsection, we denote a $k$-clustering by the set of its cluster centers.
Let $S$ be a fixed $k$-clustering, and let $O$ be the optimal clustering.
For a subset $Y \subseteq X$, we use \EMPH{$\cost(Y)$} and \EMPH{$\cost^*(Y)$} to denote $\cost(Y, S)$ and $\cost(Y, O)$, respectively.
Similarly, for a point $p \in X$, we use \EMPH{$\nn(p)$} and \EMPH{$\nn^*(p)$} to denote the nearest neighbor of $p$ in $S$ and in $O$, respectively; define \EMPH{$\dist(p)$} to be $\dist(p, S)$ and \EMPH{$\dist^*(p)$} to be $\dist(p, O)$. We partition $X$ into four subsets as follows:


%

\begin{itemize}\itemsep=0pt
\item $\EMPH{$X_{00}$} \coloneqq \Set{\big. p \in X \mid \nn(p) \in S \setminus O, \nn^*(p) \in O \setminus S }$;
\item $\EMPH{$X_{01}$} \coloneqq \Set{\big. p \in X \mid \nn(p) \in S \setminus O, \nn^*(p) \in S\cap O }$;
\item $\EMPH{$X_{10}$} \coloneqq \Set{\big. p \in X \mid \nn(p) \in S \cap O, \nn^*(p) \in O\setminus S }$;
\item $\EMPH{$X_{11}$} \coloneqq \Set{\big. p \in X \mid \nn(p) \in S \cap O, \nn^*(p) \in S\cap O }$.
\end{itemize}
Observe that for any point $p$ in $X_{11}$, $\nn(p) = \nn^*(p)$ and $\cost(p) = \cost^*(p)$;
for any point $p$ in $X_{01}$, one has $\cost(p) \le \cost^*(p)$;
and for any point $p$ in $X_{10}$, one has $\cost(p) \ge \cost^*(p)$.
Costs $\dist(p)$ and $\dist^*(p)$ are not directly comparable for point $p$ in $X_{00}$.
A $k$-clustering $S$ is \EMPH{$C$-good} for some parameter $C \geq 0$ if $\cost(X) \le \cost^*(X) + C \cdot \cost^*(X_{00})$.

%

\begin{lemma}
    \label{L:ls-good-opt}
Any $C$-good clustering $S$ for an $\alpha$-stable clustering instance $(X,\dist, \cost)$ must be optimal for $\alpha \ge C+1$.
\end{lemma}


\begin{proof}
Define a perturbed distance function $\EMPH{$\tilde\dist$}\colon X \times X \to \reals_{\ge 0}$ with respect to the given clustering $S$ as follows:
\[
\tilde\dist(p',p) \coloneqq \begin{cases}
\alpha \cdot \dist(p',p) & \text{if $p\neq \nn(p')$,} \\
\dist(p',p) & \text{otherwise.}
\end{cases}
\]


Note that $\tilde\dist$ is not symmetric. Let \EMPH{$\tilde\cost(\cdot,\cdot)$} denote the cost function under the perturbed distance function $\tilde\dist$.
The optimal clustering under perturbed cost function is the same as the original optimal clustering $O$ by the stability assumption.  Since $\nn(p) = \nn^*(p)$ if and only if $p \in X_{11}$, the cost of $O$ under the perturbed cost can be written as:
\[
\tilde\cost(X,O) =
\alpha \cdot \cost(X_{00},O) +
\alpha \cdot \cost(X_{01},O) +
\alpha \cdot \cost(X_{10},O) +
\cost(X_{11},O).
\]

By definition of perturbed distance $\tilde\dist$, $\tilde\cost(X,S) = \cost(X,S)$.
Now, by the assumption that clustering $S$ is $C$-good,
\begin{align*}
\tilde\cost(X,S) = \cost(X,S) &\le \cost(X,O) + C \cdot \cost(X_{00},O) \\
&\le
(C+1) \cdot \cost(X_{00},O) +
\cost(X_{01},O) +
\cost(X_{10},O) +
\cost(X_{11},O) \\
&\le
\tilde\cost(X,O);
\end{align*}
the last inequality follows by taking $\alpha \ge C+1$.
This implies that $S$ is an optimal clustering for $(X,\tilde\dist)$, and thus is equal to $O$.
\end{proof}

%
%




\begin{figure}[t]
\centering
\includegraphics[width=0.4\textwidth,scale=0.9]{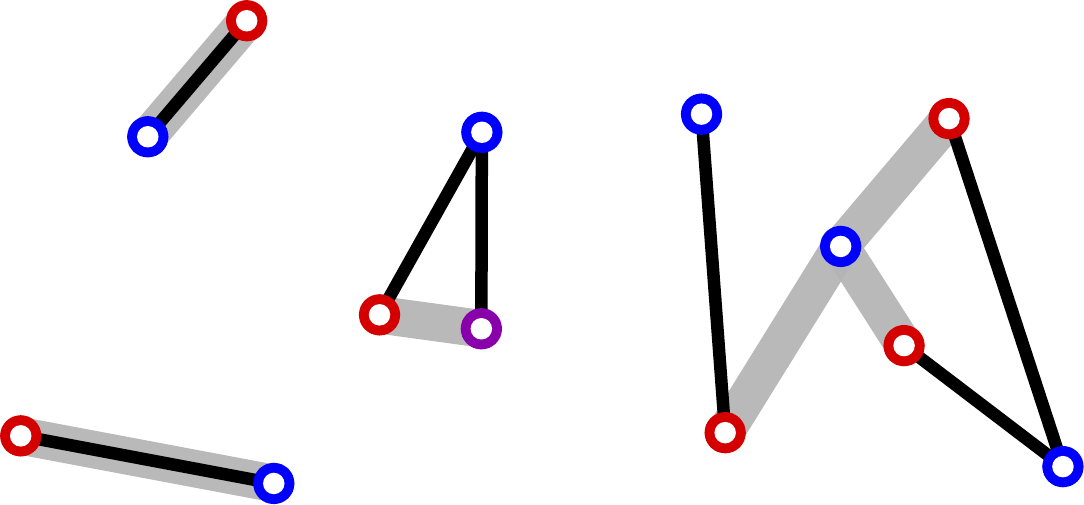}
\caption{Illustration of candidate swaps $\mathcal{S}$ in $\reals^2$.  The blue dots belong to set $S$, the red dots belong to set $O$; the only purple dot is in $S \cap O$.  The thick gray segments indicate pairs inside the stars; each star has exact one blue dot as its center.  The black pairs are the candidate swaps.  Notice that the partitions of $S$ and $O$ form connected components.}
\label{F:candidate-swaps}
\end{figure}

Next, we prove a lower bound on the improvement in the cost of a clustering that is not $C$-good after performing a $1$-swap.
Following Arya~\etal~\cite{agk+-lshkf-2004}, define the set of \EMPH{candidate swaps $\mathcal{S}$} as follows:
For each center $i$ in $S$, consider the \EMPH{star $\Sigma_i$} centered at $i$ defined as the collection of pairs
\(
\Sigma_i \coloneqq \set{ (i,j) \in S \times O \mid \nn(j) = i}.
\)
Denote \EMPH{$\Center(j)$} to be the center of the star where $j$ belongs; in other words, $\Center(j) = i$ if $j$ belongs to $\Sigma_i$.
%

For $i \in S$, let $\EMPH{$O_i$} \coloneqq \{ j \in O \ | \ \Center(j) = i \}$ be the set of centers of $O$ in star~$\Sigma_i$.
If $|O_i| = 1$, then we add the only pair $(i, j) \in \Sigma_i$ to the candidate set $\mathcal{S}$.
Let $\EMPH{$S_\varnothing$} \coloneqq \{ i\in S\  |\  O_i = \varnothing \}$.
Let $\EMPH{$O_{>1}$}$ contain centers in $O$ that belong to a star of size greater than $1$.
We pick $|O_{>1}|$ pairs from $S_\varnothing \times O_{>1}$ such that each point of $O_{>1}$ is matched only once and each point of $S_\varnothing$ is matched at most twice and add them to $\mathcal{S}$; this is feasible because $|S_\varnothing| \geq |O_{>1}|/2$.
Since each center in~$O$ belongs to exactly one pair of $\mathcal{S}$, $|\mathcal{S}| = k$.
By construction, if $|\Sigma_i| \geq 2$, then $i$ does not belong to any candidate swap.
See Figure~\ref{F:candidate-swaps}.

\begin{lemma}
\label{L:1swap-candidate}
For each point $p$ in $X_{01}$, $X_{10}$, or $X_{11}$, the set of candidate swaps $\mathcal{S}$ satisfies
\begin{align}
\label{Eq:avg1}
\sum_{(i,j) \in \mathcal{S}} (\dist(p) - \dist'(p)) \ge  \dist(p) - \dist^*(p);
\end{align}
and for each point $p$ in $X_{00}$, the set of candidate swaps $\mathcal{S}$ satisfies
\begin{align}
\label{Eq:avg2}
\sum_{(i,j) \in \mathcal{S}} (\dist(p) - \dist'(p)) \ge  (\dist(p) - \dist^*(p)) - 4 \dist^*(p),
\end{align}
where $\cost'$ is the cost function on $X$ defined with respect to $S' \coloneqq S - i + j$, and $\dist'(p)$ is the distance between $p$ and its nearest neighbor in $S'$.
\end{lemma}

\begin{proof}
For point $p$ in $X_{11}$, both $\nn(p)$ and $\nn^*(p)$ are in $S'$, so $\delta'(p) = \delta(p) = \delta^*(p)$.
%
For point $p$ in $X_{01}$, $\dist(p) \le \dist^*(p)$; when $\nn(p)$ is being swapped out by some in 1-swap $S'$, $\nn^*(p)$ must be in $S'$.
%
For point $p$ in $X_{10}$,  $\dist(p) \ge \dist^*(p)$; center $\nn(p)$ will never be swapped out by any 1-swap in $\mathcal{S}$, so $\delta'(p) \leq \delta(p)$. By construction of $\mathcal{S}$, there is exactly one choice of $S'$ that swaps $\nn^*(p)$ in; for that particular swap we have $\dist'(p) = \dist^*(p)$.
In all three cases one has inequality~(\ref{Eq:avg1}).
Our final goal is to prove inequality~(\ref{Eq:avg2}).
Consider a swap $(i, j)$ in $\mathcal{S}$. There are three cases to consider:
\begin{itemize}\itemsep=0pt
    \item \emph{$j = \nn^*(p)$}.
    There is exactly one swap for which $j = \nn^*(p)$. In this case $\dist(p) \leq \dist^*(p)$, therefore $\dist(p) - \dist'(p) \geq \dist(p) - \dist^*(p)$.
    \item \emph{$j \neq \nn^*(p)$ and $i \neq \nn(p)$}.
    Since $\nn(p) \in S'$, $\dist'(p) \leq \dist(p)$. Therefore $\dist(p) - \dist'(p) \geq 0$.
    \item \emph{$j \neq \nn^*(p)$ and $i = \nn(p)$}.
    By construction, there are most two swaps in $\mathcal{S}$ that may swap out $\nn(p)$.
    We claim that $i \neq \Center(\nn^*(p))$. Indeed, if $i = \Center(\nn^*(p))$, then by construction, $\Sigma_i = \{ (i, \nn^*(p))\}$ because the center of star of size greater than one is never added to a candidate swap.
    But this contradicts the assumption that $j \neq \nn^*(p)$.
    The claim implies that $\Center(\nn^*(p)) \in S'$ and thus $\dist'(p) \leq \dist(p, \Center(\nn^*(p)))$.
    We obtain a bound on $\dist(p, \Center(\nn^*(p)))$ as follows:
    \begin{align*}
   \dist(p, \Center(\nn^*(p))) &\le \dist(p, \nn^*(p)) + \dist( \nn^*(p), \Center(\nn^*(p)) ) \\
   &\le \dist^*(p) + \dist(\nn^*(p), \nn(p)) \le \dist^*(p) + (\dist^*(p) + \dist(p)) = \dist(p) + 2\dist^*(p).
   \end{align*}
    Therefore, $\dist(p) - \dist'(p) \geq \dist(p) - \dist(p, \Center(\nn^*(p)))$. Putting everything together, we obtain:
    \begin{align*}
        \sum_{S' \in \mathcal{S}} (\dist(p) - \dist'(p)) \ge (\dist(p) - \dist^*(p)) + 0 + 2(\dist(p) - \dist(p) - 2\dist^*(p)) = \dist(p) - 5\dist^*(p). 
    \end{align*}
\end{itemize}
\end{proof}

Using Lemma~\ref{L:1swap-candidate}, we can prove the following.
\begin{lemma}
\label{L:1swap-improv}
Let $S$ be a $k$-clustering of $(X,\dist)$ that is not $C$-good for some fixed constant $C > 4 + \e$ with arbitrarily small $\e > 0$.
There is always a $1$-swap $S'$ such that
$\cost'(X) - \cost^*(X) \le (1-\e/(1+\e)k) \cdot (\cost(X) - \cost^*(X))$,
where $\cost'$ is the cost function defined with respect to $S'$.
\end{lemma}


\begin{proof}
By Lemma~\ref{L:1swap-candidate} one has $\cost(X) - \cost'(X) \ge (\cost(X) - \cost^*(X) - \Psi(X_{00})) / k$ for some 1-swap $S'$ and its corresponding cost function $\cost'(\cdot)$.
Since $S$ is not $C$-good, $\cost(X) - \cost^*(X) > C\cdot \cost^*(X_{00})$.
Rearranging and plugging the definition of $\Psi(\cdot)$, we have
\begin{align*}
\cost'(X) - \cost^*(X)
&\le \cost(X) - \cost^*(X) - (\cost(X) - \cost^*(X) - \Psi(X_{00})) / k \\
&\le \cost(X) - \cost^*(X) - \Paren{\cost(X) - \cost^*(X) - 4\cdot \cost^*(X_{00})} / k \\
&\le \cost(X) - \cost^*(X) - \Paren{\cost(X) - \cost^*(X) + (M-1)\cdot (\cost(X) - \cost^*(X)) - 4M\cdot \cost^*(X_{00})} / Mk \\
&\le \Paren{1-\frac{\e}{(1+\e)k}} \cdot (\cost(X) - \cost^*(X)),
\end{align*}
where the last inequality holds by taking $M$ to be arbitrarily large (say $M > 1+1/\e$).
\end{proof}

\section{Efficient Implementation of Local Search}
\label{S:efficient-cost}

We describe an efficient implementation of each step of the local-search algorithm in this section.
%
By Lemma~\ref{L:change-metric}, it suffices to implement the algorithm using a polyhedral metric $\dist_N$.
We show that each step of $1$-swap can be implemented in $O(nk^{2d - 1} \polylog n)$ time under the assumption that $\alpha > 5$.
%
We obtain the following:
\begin{theorem}
    \label{thm:cost-highdim}
    Let $(X, \dist)$ be an $\alpha$-stable instance of the $k$-median problem where $X \subset \reals^d$ and $\dist$ is the Euclidean metric.
    For $\alpha > 5$, the $1$-swap local search algorithm computes the optimal $k$-clustering of $(X, \dist)$ in $O(n k^{2d - 1} \polylog n)$ time.
\end{theorem}

For simplicity, we present a slightly weaker result for $d = 2$ using the $L_1$-metric, as it is straightforward to implement and more intuitive.
Using the $L_1$-metric requires $\alpha > 5\sqrt{2}$.
The extension to higher dimensional Euclidean space using the polyhedral works for $\alpha > 5$.

\begin{figure}[t]
\centering
\includegraphics[width=0.3\textwidth, trim={0, 0.5in, 0, 0.7in}, clip]{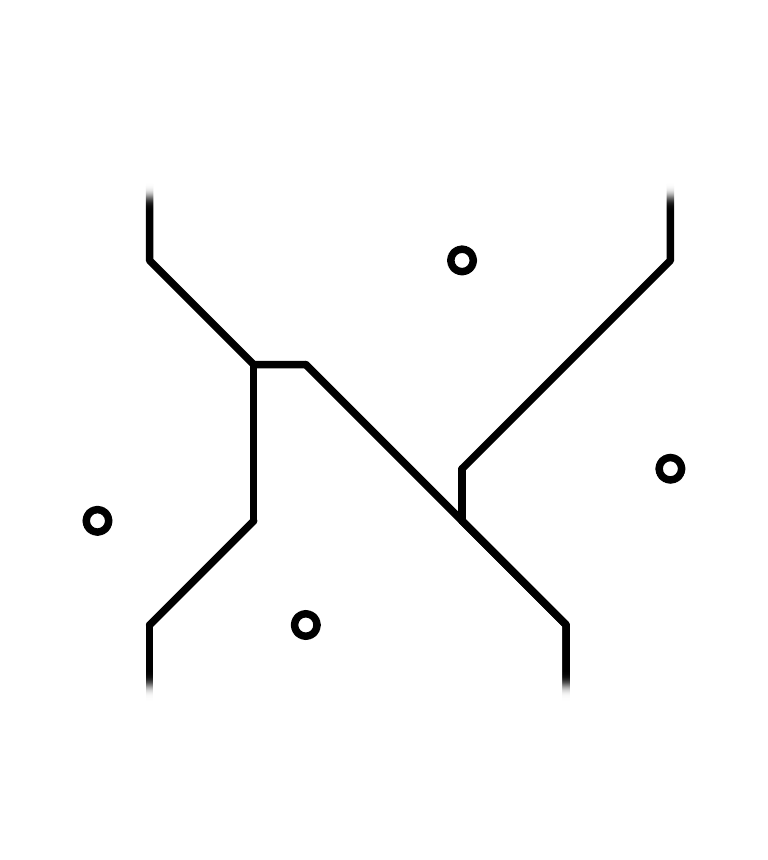}
\includegraphics[width=0.3\textwidth, trim={0, 0.5in, 0, 0.7in}, clip]{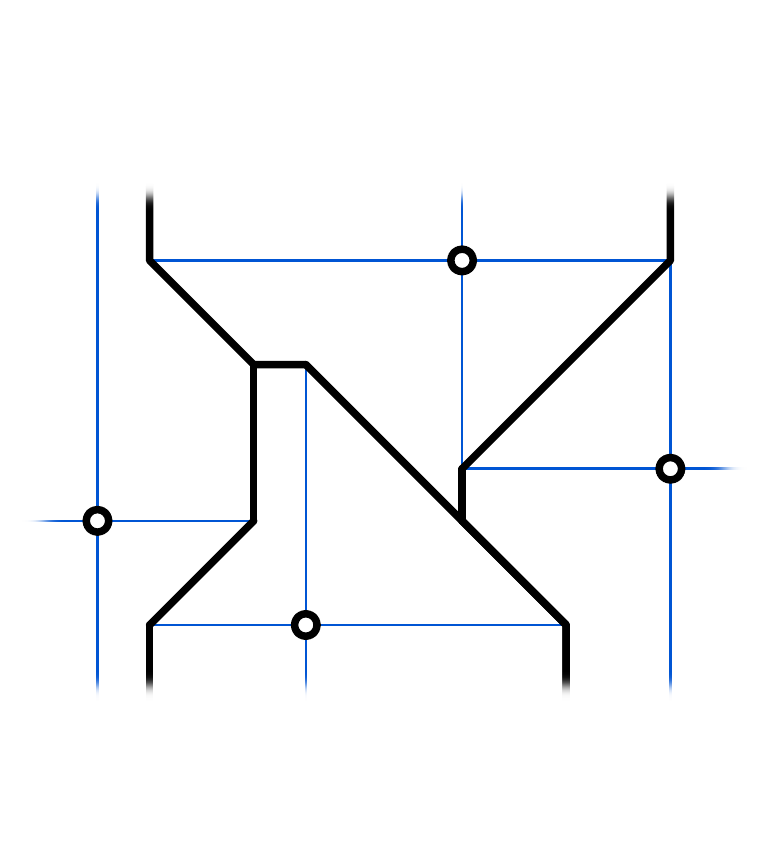}
\includegraphics[width=0.3\textwidth, trim={0, 0.5in, 0, 0.7in}, clip]{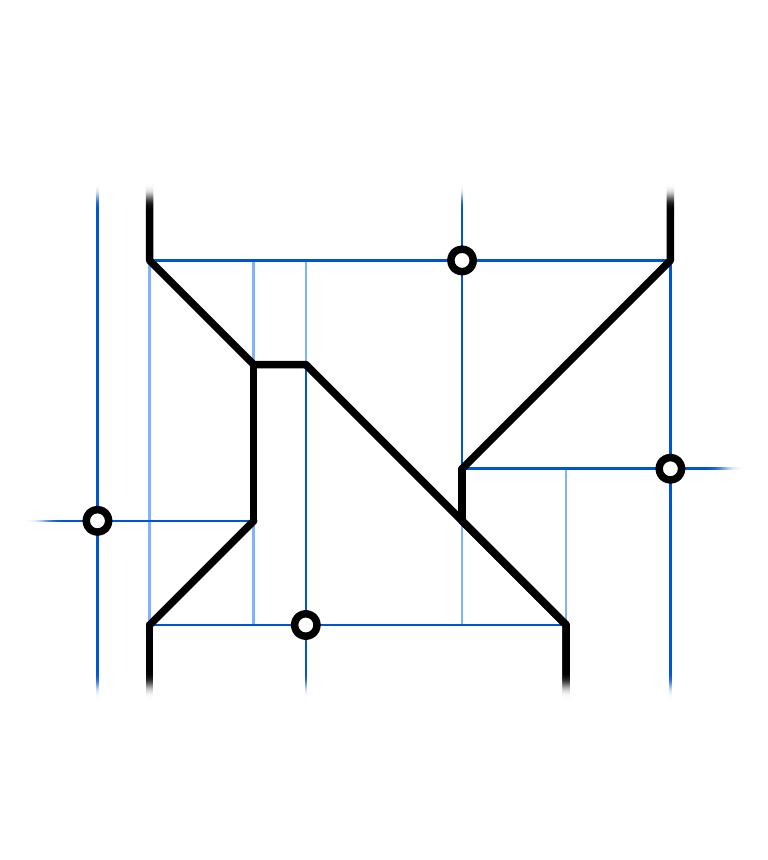}
\caption{$L_1$ Voronoi diagram $V$, quadrant decomposition $\tilde{V}$, and trapezoid decomposition $V^{\parallel}$.}
\label{F:l1-trapezoids}
\end{figure}

\subparagraph*{Voronoi diagram under {$\mathbf{L_1}$} norm.}
First, we fix a point $x \in X \setminus S$ to insert and a center $y \in S$ to drop.
Define $S' \coloneqq S + x - y$.
We build the $L_1$ Voronoi diagram ${V}$ of $S'$.
The cells of ${V}$ may not be convex, but they are \emph{star-shaped}: for any $c \in S'$ and for any point $x \in \Vor(c)$, the segment $c x$  lies completely in $\Vor(c)$.
Furthermore, all line segments on the cell boundaries of $V$ must have slopes belonging to one of the four possible values: vertical, horizontal, diagonal, or antidiagonal.

Next, decompose each Voronoi cell $\Vor(c)$ into four \emph{quadrants} centered at $c$.
Denote the resulting subdivision of $V$ as $\tilde{V}$.
We compute a \emph{trapezoidal decomposition} $V^{\parallel}$ of the diagram $\tilde{V}$ by drawing a vertical segment from each vertex of $\tilde{V}$ in both directions until it meets an edge of $V$; $V^{\parallel}$ has $O(k)$ trapezoids, see Figure~\ref{F:l1-trapezoids}.
For each trapezoid $\tau \in V^{\parallel}$, let $X_\tau \coloneqq X \cap \tau$.
The cost of the new clustering $S'$ can be computed as $\cost(X, S') = \sum_{\tau \in V^{\parallel}} \cost(X_\tau, S')$.



\subparagraph*{Range-sum queries.}
Now we discuss how to compute $\cost(X_\tau, S')$.
Each trapezoid $\tau$ in cells $\Vor(c)$ is associated with a vector $u(\tau) \in \set{\pm 1}^2$, depending on which of the four quadrants $\tau$ belongs to with respect to the axis-parallel segments drawn passing through the center $c$ of the cell.
If $\tau$ lies in the top-right quadrant then $u(\tau) = (1, 1)$.
Similarly if $\tau$ lies in the top-left (resp.\ bottom-left, bottom-right) then $u(\tau) = (-1, 1)$ (resp.\ $(-1, -1)$, $(1, -1)$).

\begin{equation}
    \cost(X_\tau, S') = \sum_{x \in X_\tau} \norm{x - c}_1 = \sum_{x \in X_\tau} \seq{x - c, u(\tau)} = \sum_{x \in X_\tau} \seq{x , u(\tau)} - { |X_\tau| \cdot \seq{c, u(\tau)} }.
\end{equation}

We preprocess $X$ into a data structure
that answers the following query:
\begin{itemize}
\item \textsc{TrapezoidSum}$(\tau, u)$: Given a trapezoid $\tau$ and a vector $u \in \set{\pm 1}^2$, return $|X \cap \tau|$ as well as $\sum_{x \in X \cap \tau} \seq{x , u}$.
\end{itemize}

The above query can be viewed as a $3$-oriented polygonal range query~\cite{de1997computational}.
We construct a $3$-level range tree $\Psi$ on~$X$.
Omitting the details (which can be found in~\cite{de1997computational}), $\Psi$ can be constructed in $O(n \log^2 n)$ time and uses $O(n \log^2 n)$ space.
Each node $\xi$ at the third level of $\Psi$ is associated with a subset $X_\xi \subseteq X$.
We store $w(\xi, u) \coloneqq \sum_{x \in X_\xi} \seq{x, u}$ for each $u \in \set{\pm 1}^2$ and $|X_\xi|$ at $\xi$.
For a trapezoid $\tau$, the query procedure identifies in $O(\log^3 n)$ time a set $\Xi_\tau$ of $O(\log^3 n)$ third-level nodes such that $X \cap \tau = \cup_{\xi \in \Xi_\tau} X_\xi$ and each point of $X\cap \tau$ appears as exactly one node of $\Xi_\tau$.
Then $\sum_{x \in X_\tau} \seq{x, u} = \sum_{\xi \in \Xi_\tau} w(\xi, u)$ and  $|X_\tau| = \sum_{\xi \in \Xi_\tau} |X_\xi|$.

With the information stored at the nodes in $\Xi_\tau$, \textsc{TrapezoidSum}$(\tau, u)$ query can be answered in $O(\log^3 n)$ time.
By performing \textsc{TrapezoidSum}$(\tau, u(\tau))$ query for all $\tau \in V^{\parallel}$, $\cost(X_\tau, S')$ can be computed in $O(k \log^3 n)$ time since $V^{\parallel}$ has a total of $O(k)$ trapezoids.

\begin{figure}[t]
\centering\small
\begin{algorithm}
\textul{$\textsc{$1$-Swap}(X, S)$:}\+
\\	\emph{input:} Point set $X$ and centers $S$
\\  for each point $x \in X \setminus S$ and center $y \in S$:\+
\\    $S' \gets S + x - y$
\\    ${V} \gets$ $L_1$ Voronoi diagram of $S'$
\\    ${\tilde{V}} \gets$ decompose each cell $\Vor(c)$ into four quadrants centered at $c$
\\    ${V^{\parallel}} \gets$ trapezoidal decomposition of $\tilde{V}$
\\    for each trapezoid $\tau \in V^{\parallel}$:\+
\\      $\cost(X_\tau, S') \gets \mathsc{TrapezoidSum}(\tau, u(\tau))$\-
\\    $\cost(X, S') \gets \sum_{\tau \in {V^{\parallel}}} \cost(X_\tau, S')$\-
\\  return $(x, y)$ with the lowest $\cost(X, S + x - y)$
\end{algorithm}
\caption{Efficient implementation of $1$-swap under $1$-norm.}
\label{F:1-swap}
\end{figure}

We summarize the implementation of $1$-swap algorithm in Figure~\ref{F:1-swap}.
The $1$-swap procedure considers at most $nk$ different $k$-clusterings.
Therefore we obtain the following.

\begin{lemma} \label{lem:cost-runtime}
Let $(X, \dist, \cost)$ be a given clustering instance where $\dist$ is the $L_1$ metric, and let $S$ be a given $k$-clustering.
After $O(n \log n)$ time preprocessing,
we find a $k$-clustering $S' \coloneqq S + x - y$ minimizing $\cost(X, S')$ among all choices of $(x, y)$ in $O(nk^2 \log^3 n)$ time.
\end{lemma}

\subsection{Cost of $\mathbf{1}$-swap in higher dimensions}
\label{S:appendix-cost-highd}
The \textsc{$1$-swap} algorithm can be extended to higher dimensions using the theory of geometric arrangements~\cite{sa-dssga-1995,as-aa-2000,aps-sugo-2008}.
The details are rather technical, so we only sketch the proofs here.
As in Section~\ref{SS:1-clustering}, instead of working with the $L_1$ metric, we work with a polyhedral metric.
Let the centrally-symmetric set $N \subseteq S^{d - 1}$ and the convex polyhedron $Q$ be defined as in Section~\ref{SS:1-clustering}.
The set $N$ partitions $\reals^d$ into a set of $O(1)$ polyhedral cones denoted by $\mathcal{C} \coloneqq \Set{ C_1, \dots, C_\gamma }$, each with $0$ as its apex so that all points (when viewed as vectors) in a cone have the same vector $u$ of $N$ as the nearest neighbor under the cosine distance, i.e.\ the polyhedral distance $\dist_N(0, u)$ is realized by $u$.
The total complexity of $\mathcal{C}$ is $O(\gamma) = O(1)$.

We show that $X$ can be preprocessed into a data structure so that $\cost(X, S)$, the cost of the $k$-clustering induced by any $k$-point subset $S$ of $X$ under $\dist_N$ can be computed in $O(k^{2d} \polylog(n))$ time.

\def\VD{\mathrm{VD}}

Let $S \subset X$ be a set of $k$ points.
We compute the Voronoi diagram $\VD(S)$ of $S$ under the distance function $\dist_N$.
More precisely, for a point $c \in S$, let $f_c \colon \reals^d \rightarrow \reals_{\geq 0}$ be the function $f_c(x) \coloneqq \dist_N(c, x)$.
The graph of $f_c$ is a polyhedral cone in $\reals^{d+1}$ whose level set for the value $\lambda$ is the \emph{homothet} copy of $Q$, $c + \lambda Q$.
Voronoi diagram $\VD(S)$ is the minimization diagram of function $f_c$ over every point $c$ in $S$; that is, the projection of the lower envelope $f(x) \coloneqq \min_{c} f_c(x)$ onto the hyperplane $X_{d+1} = 0$ (identified with $\reals^d$).
We further decompose each Voronoi cell $\Vor(c)$ of $\VD(S)$ by drawing the family of cones in $\mathcal{C}$ from $c$; put it differently, by drawing the cone $c + C_j$ for $1 \leq j \leq k$, within the cell $\Vor(c)$.
Each cell $\tau$ in the refined subdivision of $\Vor(c)$ has the property that for all points $x \in \tau$, $\dist_N(x, c)$ is realized by the vector of $N$---by $u_j$ if $x \in c + C_j$.
Let $\tilde{V}$ denote the resulting refinement of $\VD(S)$.

Finally, we compute the vertical decomposition of each cell in $\tilde{V}$, which is the extension of the trapezoidal decomposition to higher dimensions; see \cite{sa-dssga-1995,k-atubv-2004} for details.
Let $V^\parallel$ denote the resulting convex subdivision of $\reals^d$.
It is known that $V^\parallel$ has $O(k^{2d - 2})$ cells, that it can be computed in $O(k^{2d - 2})$ time, and that each cell of $V^\parallel$ is convex and bounded by at most $2d$ facets, namely it is the intersection of at most $2d$ halfspaces.
%
Using the same structure of the distance function $\dist_N$, we can show that there is a set {$U$} of $O(\gamma^d) = O(1)$ unit vectors such that each facet of a cell in $V^\parallel$ is normal to a vector in $U$, and that $U$ depends only on $N$ and not on $S$.

With these observations at hand, we preprocess $X$ into a data structure as follows: we fix a $2d$-tuple $\bar{u} \coloneqq (u_1, \dots, u_{2d}) \in U^{2d}$.
Let $R_{\bar{u}}$ be the set of all convex polyhedra formed by the intersection of at most $2d$ halfspaces each of which is normal to a vector in $\bar{u}$.
Using a multi-level range tree (consisting of $2d$ levels), we preprocess $X$ in $O(n \log^{2d} n)$ time into a data structure $\Psi_{\bar{u}}$ of size $O(n\log^{2d-1} n)$ for each $\bar{u}$, so that for a query cell $\tau \in R_{\bar{u}}$ and for a vector $u \in N$, we can quickly compute the total weight
$w(\tau, u) = \sum_{p \in X \cap \tau} \seq{p, u}$ in $O(\log^{2d} n)$ time.

For a given $S$, we compute the cost $\cost(X, S)$ as follows.
We first compute $\VD(S)$ and $V^\parallel(S)$.
For each cell $\tau \in V^{\parallel}(S)$ lying in $\Vor(c)$, let $u(\tau) \in N$ be the vector $u_j$ such that $\tau \subseteq c + C_j$. As in the 2d case,

\begin{align*}
\cost(X, S) = \sum_{c} \sum_{p \in \Vor(c)} \dist_N(p, c)
&= \sum_{c} \sum_{\tau \in V^\parallel(S) \cap \Vor(c)} \,\, \sum_{p \in X \cap \tau} \seq{p - c, u(\tau)} \\
&= \sum_{c} \sum_{\tau \in V^\parallel(S) \cap \Vor(c)} \Paren{ \sum_{p \in X \cap \tau} \seq{p, u(\tau)} - |X \cap \tau| \cdot \seq{c, u(\tau)} }.
\end{align*}

Fix a cell $\tau \in V^\parallel(S) \cap \Vor(c)$.
Suppose $\tau \in R_{\bar{u}}$.
Then by querying the data structure $\Psi_{\bar{u}}$ with $\tau$ and $u(\tau)$, we can compute $w(\tau, u) = \sum_{p \in X \cap \tau} \seq{p, u(\tau)}$ in $O(\log^d n)$ time.
Repeating this procedure over all cells of $V^\parallel(S)$, $\cost(X, S)$ can be computed in $O(k^{2d -1} \log^{2d} n)$ time, after an initial preprocessing of $O(n \log^{2d} n)$ time.

\section{Coresets and an Alternative Linear Time Algorithm}
\label{S:coresets}

In this section we provide an alternative way to compute the optimal $k$-clustering, where the objective can be any of $k$-center, $k$-means, or $k$-median.
Here we are aiming for a running time linear in $n$, but potentially with exponential dependence on $k$.
With such a goal we can further relax the stability requirement using the idea of \emph{coresets}.
When there is strict separation between clusters (when $\alpha \ge 2 + \sqrt{3}$), we can recover the optimal clustering. We note that this provides a significant improvement to the stability parameter needed for $k$-median over the local search approach, albeit with worse running time dependence on $k$.

\subparagraph*{Coresets.} Let $(X,\dist)$ be a clustering instance.
The \EMPH{radius} of a cluster $X_i$ is the maximum distance between its center and any point in $X_i$.
Let $S$ be a given $k$-clustering of $(X,\dist)$, with clusters $X_1,\dots,X_k$, centers $c_1, \dots, c_k$, and radius $r_1,\dots,r_k$, respectively.
Let $O$ be the optimal $k$-clustering of $(X,\dist)$, with clusters $X^*_1,\dots,X^*_k$, centers $c^*_1, \dots, c^*_k$ and radius $r^*_1,\dots,r^*_k$, respectively.
Let $B(c, r)$ denote the ball centered at $c$ with radius $r$ under $\dist$.

%
A point set $Q \subseteq X$ is a \EMPH{multiplicative $\e$-coreset} of $X$ if every $k$-clustering $S$ of $Q$ satisfies
\[
X \subseteq \bigcup_i B\Paren{ c_i, (1+\e)\cdot r_i }.
\]

\begin{lemma}
\label{lem:optimal-centers}
Let $(X,\dist)$ be a $(1+\e)$-stable clustering instance with optimal $k$-clustering $O$.
A multiplicative $\e$-coreset of $X$ contains at least one point from each cluster of $O$.
\end{lemma}

\begin{proof}
\label{proof:kcenter}

Let $Q$ be a multiplicative $\e$-coreset of $X$.
We start by defining a $k$-clustering $S_Q$ of $Q$.
For each point $q$ in $Q$, assign $q$ to its cluster in the optimal clustering $O$.
This results in some clustering with at most $k$ clusters.
Insert additional empty clusters to create a valid $k$-clustering $S_Q$ of $Q$.

Assume that $Q$ does not contain any points from some optimal cluster $X_i^*$ of $O$.
Consider the center point $c_i^*$ of $X_i^*$.
By the fact that $Q$ is a multiplicative $\e$-coreset, $c_i^*$ must be contained in a ball resulting from the expansion of each cluster of $S_Q$ by an $\e$-fraction of its radius.
In notation, let the cluster of $S_Q$ whose expansion covers $c_i^*$ be $X_j$, with center $c_j$ and radius $r_j$.
Then one has $\dist(c_j, c_i^*) \leq (1+\e) \cdot r_j$.

Because $S_Q$ is constructed by restricting the optimal clustering $O$ on $Q$,
cluster $X_j$ is a subset of some optimal cluster in $O$: $X_j \subseteq X_j^*$.
This implies $r_j \leq r_j^*$.
Additionally, $c_j$ and $c_i^*$ lie in different optimal clusters, as $c_j$ is in $Q$ and therefore does not lie in $X_i^*$.
So by $(1+\e)$-center proximity:
\[
\dist(c_j, c_i^*) > (1+\e)\cdot\dist(c_j, c_j^*) = (1+\e)\cdot r_j^* \geq (1+\e)\cdot r_j,
\]
contradicting to $\dist(c_j, c_i^*) \leq (1+\e) \cdot r_j$.
Therefore, $Q$ must contain at least one point from each optimal cluster.
\end{proof}


\subparagraph*{Algorithm.}
We first compute a constant approximation to the clustering problem instance $(X,\dist)$.
We then recursively construct a multiplicative coreset of size $O(k!/\e^{dk})$~\cite{har2004coresets,agarwal2005geometric}.
By taking $\e = 1$, the coreset has size $O(k!)$ and for $2$-stable instances, the multiplicative $2$-coreset $Q$ contains at least one point from every optimal cluster by Lemma~\ref{lem:optimal-centers}.
After obtaining this coreset, we then need to reconstruct the optimal clustering.
By~\cite[Corollary~9]{br-dsccl-2014}, when our instance satisfies strict separation ($\alpha \geq 2 + \sqrt{3}$), taking any point from each optimal cluster induces the optimal partitioning of $X$.
To find $k$ such points, each from a different optimal cluster, we try all possible $k$ subsets of $Q$ as the candidate $k$ centers.
For each set of centers we compute its cost (under polygonal metric) using the cost computation scheme of Section~\ref{S:efficient-cost}, then take the clustering with minimum cost.
Finally, recompute the optimal centers using any $1$-clustering algorithm on each cluster of $O$.

It is known constant approximation to any of the $k$-means, $k$-median, or $k$-center instance can be computed in $O(nk)$ time~\cite{gonzalez1985clustering} and even in $\tilde{O}(n)$ time~\cite{har2004coresets,cfs-ntasc-2019} in constant-dimensional Euclidean spaces.
Thus computing the multiplicative $2$-coreset takes $O(nk^2 + k!)$ time~\cite{har2004no}.
Using the cost computation scheme from Section~\ref{S:efficient-cost}, after $O(n \log^{2d} n)$ preprocessing time, the cost of each clustering can be computed in $\tilde{O}(k^{2d-1})$ time.
There are at most $O((k!)^k)$ possible choices of center set of size $k$.

We conclude the section with the following theorem.

\begin{theorem}
    \label{theorem:coreset-main}
    Let $X$ be a set with $n$ points lying in $\reals^d$ and $k \geq 1$ an integer.
    If the $k$-means, $k$-median, or $k$-center instance for $X$ under the Euclidean distance is $\alpha$-stable for $\alpha \geq 2 + \sqrt{3}$ then the optimal clustering can be computed in $\tilde{O}(nk^2 + k^{2d-1} \cdot (k!)^k)$ time.
\end{theorem}





\bibliographystyle{abbrv}
\bibliography{stable}



\newpage
\appendix
\section{Appendix}
\label{S:appendix}

\subsection{Stability Properties}
\label{S:stability-properties}

\paragraph*{Properties of $\mathbf{\alpha}$-center proximity.}
Let $(X,\delta)$ be a clustering instance satisfying $\alpha$-center proximity, where $\delta$ is a metric and $\alpha > 1$.
Let $X_1$ be a cluster with center $c_1$ in an optimal clustering.
Let $p,p',p'' \in X_1$ and $q \in X \setminus X_1$ with $c_2$ the center of $q$'s cluster.
Then,
\begin{enumerate}[topsep=0.3em,itemsep=0em,label={(\arabic*)}]
\item
\label{it:iSep}
$(\alpha - 1)\cdot\delta(p,c_1) < \delta(p,q)$;
\item
\label{it:iiSep}
$(\alpha - 1)\cdot\delta(p,c_1) < \delta(c_1,c_2)$;
\item
\label{it:iiiSep}
$(\alpha - 1)\cdot\delta(c_1,c_1) < (\alpha + 1) \cdot \delta(p,q)$;
\item
\label{it:ivSep}
$(\alpha - 1)\cdot\delta(p,p') < \frac{2 \alpha}{\alpha - 1} \cdot \delta(p,q)$;
\hfill
$\delta(p,p') < \delta(p,q)$ for $\alpha \ge 2 + \sqrt{3}$
\item
\label{it:vSep}
$(\alpha - 1)\cdot\delta(p',p'') < \frac{2(\alpha+1)}{\alpha -1} \cdot \delta(p,q)$.
\hfill
$\delta(p',p'') < \delta(p,q)$ for $\alpha \ge 2 + \sqrt{5}$
\end{enumerate}
\medskip

\begin{proof}
\mbox{}

\noindent \ref{it:iSep}
$\delta(p,q) \le (\alpha -1)\cdot \delta(p,c_1)$ yields the following contradiction.
\begin{eqnarray*}
\alpha\cdot \delta(q,c_2) < \delta(q,c_1) \le \delta(p,c_1) + \delta(p,q) \le \alpha\cdot \delta(p,c_1)
&\Rightarrow&
\delta(q,c_2) < \delta(p,c_1)\\
\alpha\cdot \delta(p,c_1) < \delta(p,c_2) \le \delta(q,c_2) + \delta(p,q) \le \delta(q,c_2) + (\alpha -1)\cdot \delta(p,c_1)
&\Rightarrow&
\delta(p,c_1) < \delta(q,c_2)
\end{eqnarray*}

\noindent \ref{it:iiSep}
Follows from $\alpha\cdot \delta(p,c_1) < \delta(p,c_2) \le \delta(p,c_1) + \delta(c_1,c_2)$.
\smallskip

\noindent \ref {it:iiiSep}
Follows by
$
\delta(c_1,c_2) \le \delta(c_1,p) + \delta(p,q) + \delta(q,c_2) \stackrel{\ref{it:iSep}}{<} \left(\frac{2}{\alpha - 1} + 1\right) \cdot \delta(p,q) = \frac{\alpha+1}{\alpha - 1} \cdot \delta(p,q)
$.
\smallskip

\noindent \ref {it:ivSep}
Follows by
\begin{eqnarray*}
(\alpha - 1)\cdot\delta(p,p') &\le& (\alpha - 1)\cdot \delta(p,c_1) + (\alpha - 1)\cdot \delta(p',c_1)
\stackrel{\ref{it:iSep},\ref{it:iiSep}}{<} \delta(p,q) + \delta(c_1,c_2)\\
&\stackrel{\ref{it:iiiSep}}{<}& \delta(p,q) + \frac{\alpha+1}{\alpha -1}\cdot \delta(p,q)
= \frac{2 \alpha}{\alpha - 1} \cdot \delta(p,q).
\end{eqnarray*}

\noindent \ref {it:vSep}
Follows by
\begin{eqnarray*}
(\alpha - 1)\cdot\delta(p',p'') \le (\alpha - 1)\cdot\delta(p',c_1) + (\alpha - 1)\cdot\delta(p'',c_1)
\stackrel{\ref{it:iiSep}}{<}  2\cdot \delta(c_1,c_2)
\stackrel{\ref{it:iiiSep}}{<} \frac{2(\alpha+1)}{\alpha -1} \cdot \delta(p,q).
\end{eqnarray*}
\end{proof}

\begin{proof}[Proof of Lemma~\ref{L:intra-inter}]
Let $c_1$ and $c_2$ be the centers of cluster $X_1$ and $q$'s cluster, respectively.

    (i) First we show that $\dist(c_1,c_2) < \frac{\alpha+1}{\alpha-1} \cdot \dist(p'',q)$:
\begin{align*}
\dist(p'',q)
&\ge \dist(q,c_1) - \dist(p'',c_1) \\
&\ge \dist(c_1,c_2)  - \dist(q,c_2) - \dist(p'',c_1) \\
&> \dist(c_1,c_2)  - \Paren{ \dist(q,p'') + \dist(p'',q) } / (\alpha-1).
\end{align*}
Rearranging the inequality proves the claim.

Now $\alpha \cdot \dist(p,c_1) < \dist(p,c_2) \le \dist(p,c_1) + \dist(c_1,c_2)$,
which implies that $(\alpha - 1) \cdot \dist(p,c_1) < \dist(c_1,c_2)$.
It follows that
\begin{align*}
\dist(p,p')
&\le \dist(p,c_1) + \dist(p',c_1) \\
&< 2 \cdot \dist(c_1,c_2) / (\alpha - 1) \\
&< \frac{2(\alpha+1)}{(\alpha-1)^2} \cdot \dist(p'',q) \\
&\le \dist(p'',q),
\end{align*}
where the last inequality holds when $\alpha \ge 2+\sqrt{5}$.

\smallskip

\noindent (ii) Let us assume $\delta(c_1,c_2) = 1$. We know that $\delta(p,c_2) > \alpha \cdot \delta(p,c_1)$ for all $p \in X_1$. The set of points $p$ with $\delta(p,c_2) = \alpha \cdot \delta(p,c_1)$ is known as Apollonian Circle $A_1$ (with $c_1$ inside, but not centered at $c_1$!), see Fig.~\ref{fi:Apollo}. $X_1$ must be contained inside this circle $A_1$, or sphere in higher dimensions. Similarly, there is a sphere $A_2$ enclosing $q$'s cluster (relative to $X_1$).

We take the classical fact that these are circles as given, but we want to understand the involved parameters. Of course, the circle $A_1$ has to be centered on the line $\ell$ through $c_1$ and $c_2$. Let $a$ and $b$ be the intersections of $A_1$ with $\ell$, with $b$ on the segment $c_1c_2$. $\delta(c_1,b) = \alpha \delta(c_2,b) = \alpha (1 - \delta(c_1,b))$, hence $\delta(c_1,b) = \frac{1}{\alpha +1}$.
Similarly, $\delta(c_1,a) = \alpha \delta(c_2,a) = \alpha (1 + \delta(c_1,a))$, hence $\delta(c_1,a) = \frac{1}{\alpha -1}$.
This sets the diameter of $A_1$ to $\frac{1}{\alpha +1} + \frac{1}{\alpha -1} = \frac{2\alpha}{\alpha^2-1}$, and the distance between $A_1$ and $A_2$ to $1 - 2\cdot \frac{1}{\alpha+1} = \frac{\alpha-1}{\alpha+1}$. It follows that $\delta(p',p'')/\delta(p''',q) <  \frac{2\alpha}{\alpha^2-1} / \frac{\alpha-1}{\alpha+1} = \frac{2 \alpha}{(\alpha - 1)^2}$.

\end{proof}

\begin{figure}[htb]
\centering
\colorbox{white}{\includegraphics[width=0.6\textwidth]{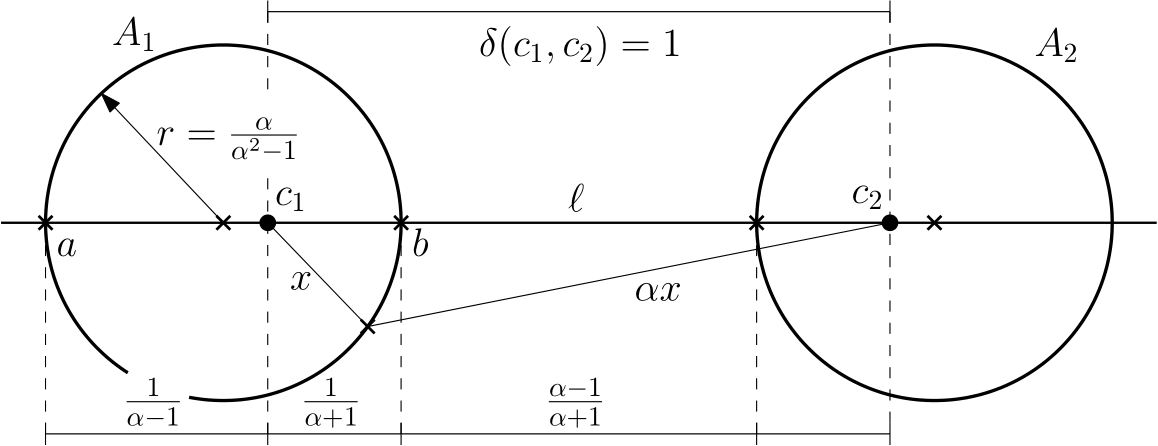}}
\caption{The Apollonian Circles (with parameter $\alpha$) for clusters centered at $c_1$ and $c_2$. }
\label{fi:Apollo}
\end{figure}

\end{document}